\documentclass[twocolumn,superscriptaddress,nofootinbib]{revtex4}
\usepackage[a4paper, left=0.8in, right=0.8in, top=0.7in, bottom=0.7in]{geometry}

\usepackage[utf8]{inputenc}
\usepackage[english]{babel}
\usepackage[T1]{fontenc}
\usepackage{amsmath}
\usepackage{bbm}
\usepackage{amsfonts}
\usepackage{hyperref}

\allowdisplaybreaks

\usepackage{tikz}
\usepackage{braket}
\usepackage{natbib}
\usepackage{amsthm}
\usepackage{amssymb}

\usepackage{comment}

\newtheorem{theorem}{Theorem}

\newtheorem{lemma}{Lemma}
\newtheorem{corollary}{Corollary}

\newtheorem{remark}{Remark}

\DeclareMathOperator*{\argmin}{arg\,min}
\DeclareMathOperator*{\sign}{sign}

\newtheorem{proposition}{Proposition}
\DeclareMathOperator{\Tr}{tr}
\DeclareMathOperator*{\E}{{\mathbb{E}}}

\DeclareMathOperator*{\Var}{\mathrm{Var}}

\DeclareMathOperator{\cD}{{\mathcal{D}}}
\DeclareMathOperator{\cE}{{\mathcal{E}}}
\DeclareMathOperator{\cF}{{\mathcal{F}}}

\newcommand{\ketbra}[2]{\lvert #1 \rangle \! \langle #2 \rvert}
\newcommand{\norm}[1]{\left\lVert#1\right\rVert}

\usepackage[noend]{algpseudocode}
\usepackage{algorithm,algorithmicx}

\algrenewcommand\alglinenumber[1]{\sf\scriptsize\color{blue}{#1}}
\algrenewcommand\algorithmicrequire{\textbf{Input:}}
\algrenewcommand\algorithmicensure{\textbf{Output:}}

\newcommand{\revision}[1]{{\color{black} #1}}
\newcommand{\newrevision}[1]{{\color{black} #1}}

\begin{document}

\title{Information-theoretic bounds on quantum advantage in machine learning}
\date{\today}
\author{Hsin-Yuan Huang}
\affiliation{Institute for Quantum Information and Matter, Caltech, Pasadena, CA, USA}
\affiliation{Department of Computing and Mathematical Sciences, Caltech, Pasadena, CA, USA}
\author{Richard Kueng}
\affiliation{Institute for Integrated Circuits, Johannes Kepler University Linz, Austria}
\author{John Preskill}
\affiliation{Institute for Quantum Information and Matter, Caltech, Pasadena, CA, USA}
\affiliation{Department of Computing and Mathematical Sciences, Caltech, Pasadena, CA, USA}
\affiliation{Walter Burke Institute for Theoretical Physics, Caltech, Pasadena, CA, USA}
\affiliation{AWS Center for Quantum Computing, Pasadena, CA, USA}

\begin{abstract}
We study the \revision{performance} 
of classical and quantum machine learning (ML) models in predicting outcomes of physical experiments.
The experiments depend on an input parameter $x$ and involve execution of a (possibly unknown) quantum process $\cE$.
Our figure of merit is the number of runs of $\cE$ \revision{required to achieve a  desired prediction performance.}
We consider classical ML models that perform a measurement and record the classical outcome after each run of $\cE$, and quantum ML models that can access $\cE$ coherently to acquire quantum data; the classical or quantum data is then used to predict outcomes of future experiments.
We prove that for any input distribution $\mathcal{D}(x)$, a classical ML model can provide accurate predictions \emph{on average} by accessing $\cE$ a number of times comparable to the optimal quantum ML model.
In contrast, for achieving accurate prediction on \emph{all} inputs, we prove that exponential quantum advantage is possible. For example, to predict expectations of all Pauli observables in an $n$-qubit system $\rho$, classical ML models require $2^{\Omega(n)}$ copies of $\rho$, but we present a quantum ML model using only $\mathcal{O}(n)$ copies.
\revision{Our results clarify where quantum advantage is possible and highlight the potential for classical ML models to address challenging quantum problems in physics and chemistry.}
\end{abstract}

\maketitle

\section{Introduction}

The widespread applications of machine learning (ML) to problems of practical interest have fueled interest in machine learning using quantum platforms \cite{biamonte2017quantum, schuld2019quantum, havlivcek2019supervised}.  Though many potential applications of quantum ML have been proposed, so far the prospect for quantum advantage in solving purely classical problems remains unclear \cite{tang2019quantum, tang2018quantum, gilyen2018quantum, arrazola2019quantum}. On the other hand, it seems plausible that quantum ML can be fruitfully applied to problems faced by quantum scientists, such as characterizing the properties of quantum systems and predicting the outcomes of quantum experiments \cite{carleo2017solving, van2017learning, carrasquilla2017machine, gilmer2017neural, melnikov2018active, sharir2020deep, aharonov2021quantum}.

Here we focus on an important class of learning problems motivated by quantum mechanics. Namely, we are interested in predicting functions of the form
\begin{equation}
f(x) = \Tr(O \cE(\lvert x \rangle \! \langle x \rvert)),
\label{eq:target-function}
\end{equation}
where $x$ is a classical input, $\cE$ is an arbitrary (possibly unknown) completely positive and trace preserving (CPTP) map, and $O$ is a known observable.
Equation~\eqref{eq:target-function} encompasses \emph{any} physical process that takes a classical input and produces a real number as output.
\revision{The goal is to construct a function $h(x)$ that accurately approximates $f(x)$ after accessing the physical process $\cE$ as few times as possible. }

A particularly important special case of setup~\eqref{eq:target-function} is training an ML model to predict what would happen in physical experiments \cite{melnikov2018active}.
Such experiments might explore, for instance, the outcome of a reaction in quantum chemistry \cite{zhou2017optimizing}, ground state properties of a novel molecule or material \cite{parr1980density, car1985unified, becke1993new, white1993density, peruzzo2014variational, kandala2017hardware, gilmer2017neural}, or the behavior of neutral atoms in an analog quantum simulator \cite{buluta2009quantum, levine2018high, bernien2017probing}. In these cases, the input $x$ subsumes parameters that characterize the process, e.g., chemicals involved in the reaction, a description of the molecule, or the intensity of lasers that control the neutral atoms. The map $\cE$ characterizes a quantum evolution happening in the lab. Depending on the parameter $x$, it produces the quantum state $\cE(\lvert x \rangle \! \langle x \rvert)$.
Finally, the experimentalist measures a certain observable $O$ at the end of the experiment.
The goal is to predict the measurement outcome for new physical experiments, with new values of $x$ that have not been encountered during the training process.

\begin{figure*}[t]
\centering
\includegraphics[width=0.83\textwidth]{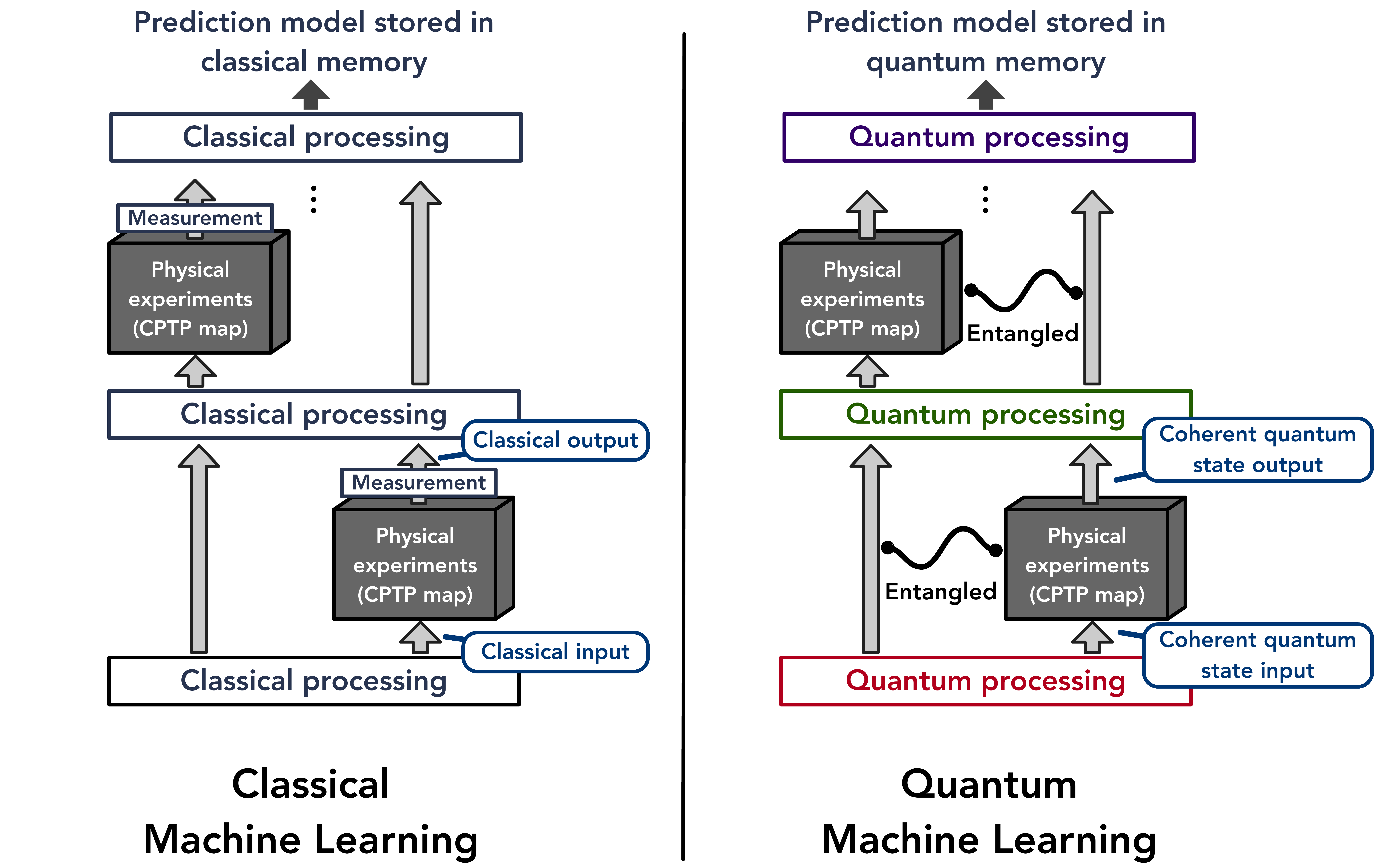}
    \caption{
    \emph{Illustration of classical and quantum machine learning settings:} The goal is to learn about an unknown CPTP map $\cE$ by performing physical experiments. (Left) In the learning phase of the classical ML setting, a measurement is performed after each query to $\cE$; the classical measurement outcomes collected during the learning phase are consulted during the prediction phase. (Right) In the learning phase of the quantum ML setting, multiple queries to $\cE$ may be included in a single coherent quantum circuit, yielding an output state stored in a quantum memory; this stored quantum state is consulted during the prediction phase.
    \label{fig:CMLvsQML}}
\end{figure*}

Motivated by these concrete applications, we want to understand the power of classical and quantum ML models in predicting functions of the form given in Equation~\eqref{eq:target-function}.
On the one hand, we consider classical ML models that can gather classical measurement data $\{(x_i, o_i)\}_{i=1}^{N_{\mathrm{C}}}$, where $o_i$ is the outcome when we perform a POVM measurement on the state $\cE(\ketbra{x_i}{x_i})$.
We denote by $N_{\mathrm{C}}$ the number of such experiments performed during training in the classical ML setting.
On the other hand, we consider quantum ML models in which 
multiple runs of the CPTP map $\cE$ can be composed coherently to collect quantum data, and predictions are produced by a quantum computer with access to the quantum data. We denote by $N_{\mathrm{Q}}$ the number of times $\cE$ is 
used during training in the quantum setting.
The classical and quantum ML settings are illustrated in Figure~\ref{fig:CMLvsQML}.

We focus on the question of whether quantum ML can have a large advantage over classical ML: to achieve a small prediction error, can the optimal $N_{\mathrm{Q}}$ in the quantum ML setting be much less than the optimal $N_{\mathrm{C}}$ in the classical ML setting? For the purpose of this comparison, we disregard the runtime of the classical or quantum ML models that generate the predictions; we are only interested in how many times the process $\cE$ must run during the learning phase in the quantum and classical settings.

Our first main result addresses small \emph{average} prediction error, i.e.\ the prediction error $|h(x)-f(x)|^2$ averaged over some specified input distribution $\mathcal{D}(x)$.
We rigorously show that, for any $\cE$, $O$, and $\mathcal{D}$, and for any quantum ML model, one can always design a classical ML model achieving a similar average prediction error such that $N_{\mathrm{C}}$ is larger than $N_{\mathrm{Q}}$ by at worst a small polynomial factor.
Hence, there is no exponential advantage of quantum ML over classical ML if the goal is to achieve a small average prediction error, and if the efficiency is quantified by the number of times $\cE$ is used in the learning process.
\revision{
This statement holds for existing quantum ML models running on near-term devices \cite{havlivcek2019supervised, schuld2019quantum, huang2020power} and future quantum ML models yet to be conceived. We note, though, that while there is no large advantage in query complexity, a substantial quantum advantage in computational complexity is possible \cite{servedio2004equivalences}.}

However, the situation changes if the goal is to achieve a small \emph{worst-case} prediction error rather than a small average prediction error --- an exponential separation between $N_{\mathrm{C}}$ and $N_{\mathrm{Q}}$ becomes possible if we insist on predicting $f(x)=\Tr(O \cE(\lvert x \rangle \! \langle x \rvert))$ accurately for \emph{every} input $x$.
We illustrate this point with an example: accurately predicting expectation values of Pauli observables in an unknown $n$-qubit quantum state $\rho$. This is a crucial subroutine in many quantum computing applications; see e.g.~\cite{peruzzo2014variational, kokail2019self, huang2019near, crawford2020efficient, huang2020predicting, izmaylov2019unitary, jiang2020optimal, huggins2019efficient}.
We present a quantum ML model that uses $N_{\mathrm{Q}}=\mathcal{O}(n)$ copies of $\rho$ to predict expectation values of all $n$-qubit Pauli observables. In contrast we prove that \emph{any} classical ML model requires $N_{\mathrm{C}}=2^{\Omega(n)}$ copies of $\rho$ to achieve the same task even if the ML model can perform arbitrary adaptive single-copy POVM measurements.

\section{Machine learning settings}

We assume that the observable $O$ (with \revision{$\norm{O} \leq 1$}) is known and the physical experiment $\cE$ is an unknown CPTP map that belongs to \revision{a set of CPTP maps $\cF$}. Apart from $\cE \in \cF$, the process can be arbitrary --- a common assumption in statistical learning theory~\cite{valiant1984theory, blumer1989learnability, bartlett2002rademacher, vapnik2013nature, arunachalam2017guest}.
For the sake of concreteness, we assume that $\cE$ is a CPTP map from a Hilbert space of $n$ qubits to a Hilbert space of $m$ qubits.
Regarding inputs, we consider bit-strings of size $n$: $x \in \left\{0,1\right\}^n$. This is not a severe restriction, since floating-point representations of continuous parameters can always be truncated to a finite number of digits.
We now give precise definitions for classical and quantum ML settings; see Fig.~\ref{fig:CMLvsQML} for an illustration.

\paragraph{Classical (C) ML:} The ML model consists of two phases: learning and prediction. During the learning phase, a randomized algorithm selects classical inputs $x_i$ and we perform a (quantum) experiment that results in an outcome $o_i$ from performing a POVM measurement on $\cE(\ketbra{x_i}{x_i})$.
A total of  $N_{\mathrm{C}}$ experiments give rise to the classical training data
 $\{(x_i, o_i)\}_{i=1}^{N_{\mathrm{C}}}$.
After obtaining this training data, the ML model 
executes a randomized algorithm $\mathcal{A}$ to learn a prediction model
\begin{equation}
    s_{\mathrm{C}} = \mathcal{A}\left(\{(x_1, o_1), \ldots (x_{N_{\mathrm{C}}},o_{N_{\mathrm{C}}})\}\right),
\end{equation}
where $s_{\mathrm{C}}$ is stored in the classical memory.
In the prediction phase, a sequence of new inputs $\tilde{x}_1, \tilde{x}_2, \ldots \in \{0, 1\}^n$ is provided. The ML model will use $s_{\mathrm{C}}$ to evaluate predictions $h_{\mathrm{C}}(\tilde{x}_1), h_{\mathrm{C}}(\tilde{x}_2), \ldots$ that approximate $f(\tilde{x}_1), f(\tilde{x}_2), \ldots$ up to small errors.

\paragraph{Restricted classical ML:} We will also consider a restricted version of the classical setting. Rather than performing arbitrary POVM measurements, we restrict the ML model to measure the target observable $O$ on the output state $\cE{\ketbra{x_i}{x_i}}$ to obtain the measurement outcome $o_i$. In this case, we always have $o_i \in \mathbb{R}$ and $\E[o_i] = \Tr(O \cE(\ketbra{x_i}{x_i}))$.

\paragraph{Quantum (Q) ML:} During the learning phase, the model starts with an initial state $\rho_0$ in a Hilbert space of arbitrarily high dimension. Subsequently, the quantum ML model accesses the unknown CPTP map $\cE$ a total of $N_{\mathrm{Q}}$ times. These queries are interleaved with quantum data processing steps:
\begin{equation} \label{eq:outputQSQML-main}
    \rho_{\cE} = \mathcal{C}_{N_{\mathrm{Q}}} (\cE \otimes \mathcal{I}) \mathcal{C}_{N_{\mathrm{Q}}-1} \ldots \mathcal{C}_1 (\cE \otimes \mathcal{I}) (\rho_0),
\end{equation}
where each $\mathcal{C}_{i}$ is an arbitrary but known CPTP map, and we write $\cE \otimes \mathcal{I}$ to emphasize that $\cE$ acts on an $n$-qubit subsystem of a larger quantum system.
The final state $\rho_{\cE}$, encoding the prediction model learned from the queries to the unknown CPTP map $\cE$, is stored in a quantum memory.
In the prediction phase, a sequence of new inputs $\tilde{x}_1, \tilde{x}_2, \ldots \in \{0, 1\}^n$ is provided. A quantum computer with access to the stored quantum state $\rho_{\cE}$ executes a computation
to produce prediction values $h_{\mathrm{Q}}(\tilde{x}_1), h_{\mathrm{Q}}(\tilde{x}_2), \ldots$ that approximate $f(\tilde{x}_1), f(\tilde{x}_2), \ldots$ up to small errors\begin{NoHyper}\footnote{Due to non-commutativity of quantum measurements, the ordering of new inputs matters. For instance, the two lists $\tilde{x}_1, \tilde{x}_2$ and $\tilde{x}_2, \tilde{x}_1$ can lead to different outcome predictions $h_{\mathrm{Q}}(\tilde{x}_i)$. Our main results do not depend on this subtletey --- they are valid, irrespective of prediction input ordering.}\end{NoHyper}.

The quantum ML setting is strictly more powerful than the classical ML setting.
During the prediction phase, classical ML models are restricted to processing classical data, albeit data obtained by measuring a quantum system during the learning phase.
In contrast, quantum ML models can work directly \revision{with the quantum data and perform quantum data processing}. A quantum ML model can have an exponential advantage relative to classical ML models for some tasks, as we demonstrate in Sec.~\ref{sec:worstcaseprederr}.

\section{Average-case prediction error}

For a prediction model $h(x)$, we consider the average-case prediction error
\begin{equation}
\sum_{x \in \{0, 1\}^n} \mathcal{D}(x) |h(x) - \Tr(O \cE(\ketbra{x}{x}))|^2,
\end{equation}
with respect to a fixed distribution $\mathcal{D}$ over inputs. This could, for instance, be the uniform distribution.

Although learning from quantum data is strictly more powerful than learning from classical data, there are fundamental limitations.  The following rigorous statement limits the potential for quantum advantage.

\begin{theorem} \label{thm:noadvquantum}
Fix an $n$-bit probability distribution $\mathcal{D}$, an $m$-qubit observable $O$ $(\norm{O}\leq 1)$ and \revision{a set $\cF$} of CPTP maps with $n$ input qubits and $m$ output qubits.
Suppose there is 
a quantum ML model which accesses the map $\cE\in \cF$ $N_{\mathrm{Q}}$ times, producing with high probability a function $h_{\mathrm{Q}}(x)$ that achieves
\begin{equation}
    \sum_{x \in \{0, 1\}^n} \mathcal{D}(x) \left| h_{\mathrm{Q}}(x) - \Tr(O \mathcal{E}(\ketbra{x}{x})) \right|^2 \leq \epsilon.
\end{equation}
Then there is an ML model in the restricted classical setting which accesses $\cE$ $N_{\mathrm{C}} =\mathcal{O}(m N_{\mathrm{Q}} / \epsilon)$ times and produces with high probability a function $h_C$ that achieves
\begin{equation} \label{eq:aveprederr}
    \sum_{x \in \{0, 1\}^n} \mathcal{D}(x) \left| h_{\mathrm{C}}(x) - \Tr(O \mathcal{E}(\ketbra{x}{x})) \right|^2 =\mathcal{O}(\epsilon).
\end{equation}
\end{theorem}
\begin{proof}[Proof sketch]
The proof consists of two parts.
First, we cover the entire set of CPTP maps $\cF$ with a maximal packing net, i.e.\ the largest subset $\mathcal{S} = \{\cE_s\}_{s=1}^{|\mathcal{S}|} \subset \cF$ such that the functions $f_{\cE_s}(x) = \Tr(O \cE_s(\ketbra{x}{x}))$ obey
$\sum_{x \in \left\{0,1\right\}^n} \mathcal{D}(x) \left| f_{\mathcal{E}_s}(x) - \newrevision{f_{\mathcal{E}_{s'}}} (x)\right|^2 >4 \epsilon$ whenever $s \neq s'$.
We then set up a communication protocol as follows. \revision{Alice chooses an element $s$ of the packing net uniformly at random, records her choice $s$, and then applies $\cE_s$ $N_Q$ times to prepare a quantum state $\rho_{\cE_s}$ as in Eq.~\eqref{eq:outputQSQML-main}. Alice's random ensemble of quantum states is thus given by
\begin{equation} \label{eq:randomensemble}
    \rho_{\cE_s} \, \mbox{ with probability } \, p_s = \tfrac{1}{|\mathcal{S}|}
\end{equation}
for $s = 1, \ldots, |\mathcal{S}|$.
Alice then sends the randomly sampled quantum state $\rho_{\cE_s}$ to Bob, hoping that Bob can decode the state $\rho_{\cE_s}$ to recover her chosen message $s$.} Using the quantum ML model,
Bob can produce the function $h_{Q, s}(x)$. Because by assumption the function $h_{Q, s}(x)$ achieves a small average-case prediction error with high probability, and because the packing net has been constructed so that the functions $\{f_{\cE_s}\}$ are sufficiently distinguishable, Bob can determine $s$ successfully with high probability. \revision{Because Alice chose from among $|\mathcal{S}|$ possible messages, the mutual information of the chosen message $s$ and Bob's measurement outcome must be at least of order $\log |\mathcal{S}|$ bits.
According to Holevo's theorem, the Holevo $\chi$ quantity of Alice's ensemble Eq.~\eqref{eq:randomensemble} upper bounds this mutual information, and therefore
must also be $\chi=\Omega(\log |\mathcal{S})|$.}
Furthermore, we can analyze how $\chi$ depends on $N_Q$, finding that each additional application of $\cE_s$ can increase $\chi$ by at most $\mathcal{O}(m)$. We conclude that $\chi =\mathcal{O}(m N_{\mathrm{Q}})$, yielding the lower bound $N_{\mathrm{Q}} = \Omega(\log(|\mathcal{S}|) / m)$. The lower bound applies to any quantum ML model, where the size $|\mathcal{S}|$ of the packing net depends on the average-case prediction error $\epsilon$. This completes the first part of the proof.

In the second part, we explicitly construct an ML model in the restricted classical setting that achieves a small average-case prediction error using a modest number of experiments.
In this ML model, an input $x_i$ is selected by sampling from the probability distribution $\mathcal{D}$, and an experiment is performed in which the observable $O$ is measured in the output quantum state $\cE(\ketbra{x_i}{x_i})$, obtaining measurement outcome $o_i$ which has expectation value $\Tr(O \cE(\ketbra{x_i}{x_i}))$.
A total of $N_{\mathrm{C}}$ such experiments are conducted.
Then, the ML model minimizes the least-squares error to find the best fit within the aforementioned maximal packing net $\mathcal{S}$:
\begin{equation}
    h_{\mathrm{C}} = \argmin_{f \in \mathcal{S}} \frac{1}{N_{\mathrm{C}}} \sum_{i=1}^{N_{\mathrm{C}}} |f(x_i) - o_i|^2.
\end{equation}
Because the measurement outcome $o_i$ fluctuates about the expectation value of $O$, it may be impossible to achieve zero training error.
Yet it is still possible for $h_{\mathrm{C}}$ to achieve a small average-case prediction error, potentially even smaller than the training error.
We use properties of maximal packing nets 
and of quantum fluctuations of measurement outcomes to perform a tight statistical analysis of the average-case prediction error, finding that with high probability $\sum_{x \in \{0, 1\}^n} \mathcal{D}(x) \left| h_{\mathrm{C}}(x) - \Tr(O \mathcal{E}(\ketbra{x}{x})) \right|^2 = \mathcal{O}(\epsilon)$, provided that $N_{\mathrm{C}}$ is of order $\log(|\mathcal{S}|) / \epsilon$.

Finally, we combine the two parts to conclude $N_{\mathrm{C}} = \mathcal{O}(m N_{\mathrm{Q}} / \epsilon)$.
The full proof is in Appendix~\ref{sec:proofnoadv}.
\end{proof}

Theorem~\ref{thm:noadvquantum} shows that all problems that are approximately learnable by a quantum ML model are also approximately learnable by some restricted classical ML model which executes the quantum process $\cE$ a comparable number of times. This 
applies in particular, to predicting outputs of quantum-mechanical processes.
The relation $N_{\mathrm{C}} = \mathcal{O}(m N_{\mathrm{Q}} / \epsilon)$ is tight. We give an example in Appendix~\ref{sec:satnoadv} with $N_{\mathrm{C}} = \Omega(m N_{\mathrm{Q}} / \epsilon)$.

\revision{For the task of learning classical Boolean circuits, fundamental limits on quantum advantage have been established in previous work \cite{servedio2004equivalences, zhang2010improved, arunachalam2016optimal, arunachalam2017guest, chung2019sample, arunachalam2020quantum}.  Theorem~\ref{thm:noadvquantum} generalizes these existing results to the task of learning outcomes of quantum processes.}

\section{Worst-case prediction error}
\label{sec:worstcaseprederr}

\begin{figure*}
\centering
\begin{minipage}{.33\textwidth}
    \renewcommand{\figurename}{Table}
    \begin{tabular}{c|c|c}
        Model & Upp. bd. & Low. bd. \\
        \hline
        Quantum ML & $\mathcal{O}(n)$ & $\Omega(n)$  \\
        Classical ML & $\mathcal{O}(n 2^n)$ & $ 2^{\Omega(n)}$\\
        Restricted CML & $\mathcal{O}(2^{2n})$ & $\Omega(2^{2n})$
    \end{tabular}
    \caption{Sample complexity for predicting expectations of \emph{all} $4^n$ Pauli observables (worst-case prediction error) in an $n$-qubit quantum state. \emph{Upp. bd.} is the achievable sample complexity of a specific algorithm. \emph{Low. bd.} is the lower bound for any algorithm. The classical ML upper bound can be achieved using classical shadows based on random Clifford measurements \cite{huang2020predicting}. The rest of the bounds are obtained in Appendix~\ref{app:exp-sep-Pauli}. }
    \label{tab:QSCC-samp}
\end{minipage}\,\,\,\,
\begin{minipage}{.64\textwidth}
    \renewcommand{\figurename}{Figure}
    \includegraphics[width=1.0\textwidth]{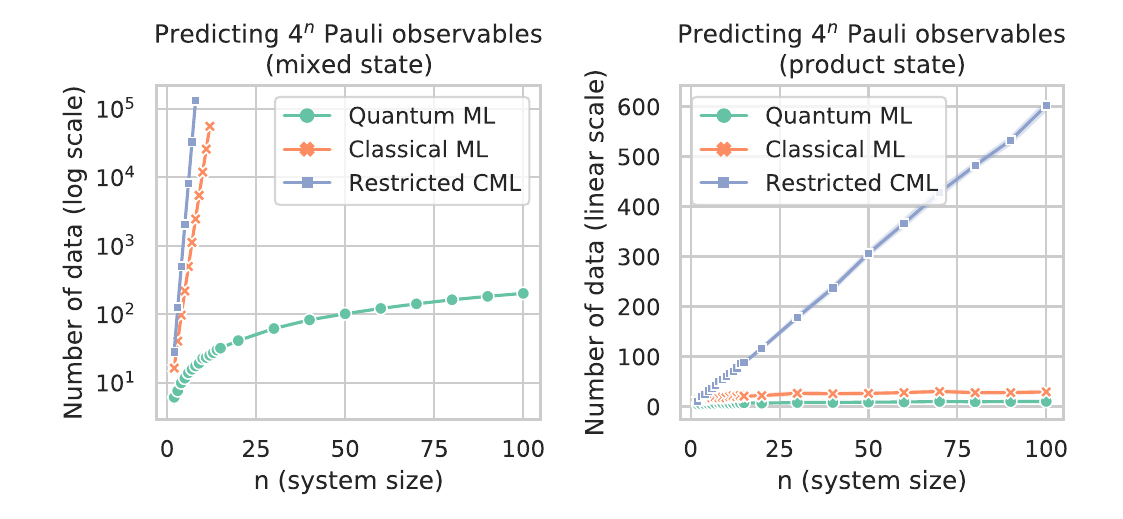}
    \caption{Numerical experiments -- number of copies of an unknown $n$-qubit state needed for predicting expectation values of all $4^n$ Pauli observables, with constant worst-case prediction error. \emph{Mixed states:} Quantum states of the form $(I + P) / 2^n$, where $P$ is an $n$-qubit Pauli observable. \emph{Product states:} Tensor products of single-qubit stabilizer states. }
    \label{fig:numeric}
\end{minipage}
\end{figure*}

Rather than achieving a small average prediction error, one may be interested in obtaining a prediction model that is accurate for all inputs $x \in \{0, 1\}^n$.
For a prediction model $h(x)$, we consider the worst-case prediction error to be
\begin{equation}
    \max_{x \in \{0, 1\}^n} |h(x) - \Tr(O \cE(\ketbra{x}{x}))|^2.
\end{equation}
Under such a stricter performance requirement, exponential quantum advantage becomes possible.

We highlight this potential by means of an illustrative and practically relevant example: predicting expectation values of Pauli operators in an unknown $n$-qubit quantum state $\rho$.
This is a central task for many quantum computing applications \cite{peruzzo2014variational, huang2019near, crawford2020efficient, huang2020predicting, kokail2019self, izmaylov2019unitary, jiang2020optimal, huggins2019efficient}.
To formulate this problem in our framework, suppose the $2n$-bit input $x$ specifies one of the $4^n$ $n$-qubit Pauli operators $P_x \in \{I, X, Y, Z\}^{\otimes n}$, and suppose that $\cE_{\rho}(|x\rangle\!\langle x|)$ prepares the unknown state $\rho$ and maps $P_x$ to the fixed observable $O$, which is then measured; hence
\begin{equation}
    f(x) = \Tr(O \cE_\rho(\ketbra{x}{x})) =
    \Tr(P_x \rho).
\end{equation}
In this setting, according to Theorem~\ref{thm:noadvquantum}, there is no large quantum advantage if our goal is to estimate the Pauli operator expectation values with a small \emph{average} prediction error.
However, an exponential quantum advantage is possible if we insist on accurately predicting \emph{every one} of the $4^n$ Pauli observables.

\revision{
First we show there is an efficient quantum ML model that achieves a small prediction error. Details are in Appendix~\ref{sec:pauli-qml}; here we just sketch the main ideas.
The procedure for predicting $\Tr(P_x \rho)$ has two stages. The goal of the first stage is to predict the absolute value $|\Tr(P_x \rho)|$ for each $x$, and the goal of the second stage is to determine the \emph{sign} of $\Tr(P_x \rho)$. The key idea used in the first stage is that, although two different Pauli operators $P_x$ and $P_y$ may either commute or anticommute, the tensor products $P_x \otimes P_x$ and $P_y \otimes P_y$ are mutually commuting for all $x$ and $y$. Therefore, although it is not possible to measure anticommuting Pauli operators simultaneously using a single copy of the state $\rho$, it is possible to measure $P_x \otimes P_x$ simultaneously for all $x$ using two copies of $\rho$. Indeed, all $4^n$ expectation values $\Tr((P_x \otimes P_x) (\rho \otimes \rho)) = \Tr(P_x \rho)^2$ can be determined by measuring pairs of qubits in the Bell basis, which is highly efficient. This completes the first stage.

If $|\Tr(P_x \rho)|$ is found to be small in the first stage, we may predict $h(x)=0$ and be assured that the prediction error is small. Therefore, in the second stage, we need only determine the sign if $|\Tr(P_x \rho)|$ was found to be reasonably large in the first stage. In that case we can perform a coherent measurement across several copies of $\rho$ which performs a majority vote and yields the correct value of the sign with high success probability.
Because the measurement is strongly biased in favor of one of the two possible outcomes, it introduces only a very ``gentle'' disturbance of the pre-measurement state.
Therefore, by performing many such measurements in succession on the same quantum memory register, we can determine the sign of $\Tr(P_x \rho)$ for many different values of $x$.
The second stage can also be more amenable to near-term implementation using a heuristic that groups commuting observables \cite{izmaylov2019unitary, huggins2019efficient}; see Appendix~\ref{sec:heuristicsNTImp} for further discussion.
Each of the two stages requires only a small number of copies of $\rho$; a careful analysis yields the following theorem.
}

\begin{theorem} \label{thm:pauli-shadow}
The quantum ML model only needs $N_{\mathrm{Q}} = \mathcal{O}(\log(M / \delta)/\epsilon^4)$
copies of $\rho$ to predict expectation values of any $M$ Pauli observables to error $\epsilon$ with probability at least $1 - \delta$.
\end{theorem}

More details regarding the quantum ML model, as well as a rigorous proof, are provided in Appendix~\ref{sec:pauli-qml}.
\revision{The sample complexity stated in Theorem~\ref{thm:pauli-shadow} improves upon previously known shadow tomography protocols \cite{aaronson2018shadow, aaronson2019gentle, huang2020predicting, badescu2020improved} for the special case of predicting Pauli observables; see Appendix~\ref{sec:related-work}.}
Because each access to $\cE_\rho$ allows us to obtain one copy of $\rho$, we only need $N_{\mathrm{Q}} = \mathcal{O}(n)$ to predict expectation values of all $4^n$ Pauli observables up to a constant error.

For classical ML models, we prove the following fundamental lower bound; see Appendix~\ref{sec:lowerboundindmeas}. 

\begin{theorem} \label{thm:cmllwbd}
Any classical ML must use $N_{\mathrm{C}} \geq 2^{\Omega(n)}$ copies of $\rho$ to predict expectation values of all Pauli observables up to a small error with a constant success probability.
\end{theorem}
This theorem holds even when the POVM measurements performed by the classical ML model could depend on the previous POVM measurement outcomes adaptively.
When combined with  Theorem~\ref{thm:pauli-shadow}, Theorem~\ref{thm:cmllwbd} establishes an exponential gap separating classical ML models from fully quantum ML models.
Table~\ref{tab:QSCC-samp} provides a summary of the upper and lower bounds on the sample complexity for predicting expectation values of Pauli observables.

\section{Numerical experiments}

We support our theoretical findings with numerical experiments,
focusing on the task of predicting the expectation values of all $4^n$ Pauli observables in an unknown $n$-qubit quantum state $\rho$, with a small worst-case prediction error.
In this case, the function is $f(x) = \Tr(O \cE_{\rho}(\ketbra{x}{x})) = \Tr(P_x \rho)$, where $x \in \{I, X, Y, Z\}^n$ indexes the Pauli observables, and $\cE_{\rho}$ prepares the unknown state $\rho$ then maps $P_x$ to the fixed observable $O$.
This is the task we considered in Section~\ref{sec:worstcaseprederr}.
Note that average-case prediction of Pauli observables is a much easier task, because most of the $4^n$ expectation values are exponentially small in $n$.

We consider two classes of underlying states $\rho$:
(i) \emph{Mixed states:}  $\rho =(I + P) / 2^n$, where $P$ is a tensor product of $n$ Pauli operators. States in this class have rank $2^{n-1}$.
(ii) \emph{Product states:} $\rho = \bigotimes_{i=1}^n |s_i \rangle \! \langle s_i|$,
where each $|s_i \rangle$ is one the six possible single-qubit stabilizer states.
We consider stabilizer states to ensure that classical simulation of the quantum ML model is tractable for reasonably large system size.

The numerical experiment in Figure~\ref{fig:numeric} implements the best-known ML procedures.
We can clearly see that there is an exponential separation between the number of copies of the state $\rho$ required for classical and quantum ML to predict expectation values when $\rho$ is in the class of mixed states.
However, for the class of product states, the separation is much less pronounced. Restricted classical ML can only obtain outcomes $o_i \in \{\pm 1\}$ with $\E[o_i] = \Tr(P_{x_i} \rho)$. Hence each copy of $\rho$ provides at most one bit of information, and therefore $\mathcal{O}(n)$ copies are needed to predict expectation values of all $4^n$ Pauli observables. In contrast, standard classical ML can perform arbitrary POVM measurements on the state $\rho$, so each copy can provide up to $n$ bits of information. The separation between classical ML and quantum ML is marginal for product states.

\vspace{-0.5em}
\section{Conclusion and outlook}
\vspace{-0.5em}

\revision{We have studied the task of learning functions of the form Equation~\eqref{eq:target-function}, using as a figure of merit the number of runs of $\cE$.
Our main result Theorem \ref{thm:noadvquantum} shows that, when the objective is achieving a specified \emph{average} prediction error, a classical ML model can perform as well as a quantum ML model, using a comparable number of runs of $\cE$.
This result establishes a fundamental limit on quantum advantage in machine learning that holds for any quantum ML model \cite{havlivcek2019supervised, schuld2019quantum, huang2020power}.

From a different perspective, Theorem \ref{thm:noadvquantum} means that the classical ML setting, in which a measurement is performed after each query to $\cE$, can be surprisingly effective. The quantum ML setting, in which multiple queries to $\cE$ can be included in a single coherent quantum circuit, is far more challenging and may be infeasible until far in the future.
Therefore finding that classical and quantum ML have comparable power (for average-case prediction) boosts our hopes that the combination of classical ML and near-term quantum algorithms \cite{preskill2018quantum, huang2019near, havlivcek2019supervised, schuld2019quantum, huang2020power} may fruitfully address challenging quantum problems in physics, chemistry, and materials science.}

On the other hand, Theorem \ref{thm:pauli-shadow}~and~\ref{thm:cmllwbd} rigorously establish that quantum ML can have an exponential advantage over classical ML for certain problems where the objective is achieving a specified \emph{worst-case} prediction error.
\newrevision{
This exponential advantage of quantum ML over classical ML may be viewed as an exponential separation between coherent measurements (in which a measurement apparatus interacts coherently multiple times with a measured system, storing quantum data which is then processed by a quantum computer) and incoherent measurements (in which a POVM measurement is performed and the outcome recorded after each interaction between system and apparatus, and the classical measurement outcomes are then processed by a classical computer).
Such a separation has been challenging to establish because incoherent measurements are difficult to analyze 
in the adaptive setting, where each measurement performed may depend on the outcomes of all previous measurements.
Our proof technique overcomes this challenge,
enabling us to identify tasks which allow substantial quantum advantage.
An important future direction will be identifying further learning problems which allow substantial quantum advantage, pointing toward potential practical applications of quantum technology.}

\vspace{-1em}
\subsection*{Acknowledgments:}
\vspace{-0.5em}
The authors thank Victor Albert, Sitan Chen, Jerry Li, Seth Lloyd, Jarrod McClean, Spiros Michalakis, Yuan Su, and Thomas Vidick for valuable input and inspiring discussions.
We would also like to thank anonymous reviewers for in-depth comments and suggestions.
HH is supported by the J. Yang \& Family Foundation.
JP acknowledges funding from  the U.S. Department of Energy Office of Science, Office of Advanced Scientific Computing Research, (DE-NA0003525, DE-SC0020290), and the National Science Foundation (PHY-1733907). The Institute for Quantum Information and Matter is an NSF Physics Frontiers Center.

\bibliography{ref}
\bibliographystyle{abbrv}

\onecolumngrid
\newpage
\appendix

\section*{\Large Appendix}

\paragraph*{\textbf{Roadmap:}}
Appendix~\ref{sec:related-work} provides additional context and discusses relevant existing work. Details regarding the numerical experiments can be found in Section~\ref{sec:numerics}. The remaining portions are devoted to theory and mathematical proofs. Appendix~\ref{sec:proofnoadv} provides a thorough treatment of average prediction errors, including bounds on query complexity (the number times the quantum process $\cE$ is accessed) culminating in a proof of Theorem~\ref{thm:noadvquantum}. Appendix~\ref{sec:satnoadv} provides a stylized example demonstrating that the bound is tight.
Finally, Appendix~\ref{app:exp-sep-Pauli} provides sample complexity upper and lower bounds for predicting many Pauli expectation values with small worst-case error (leading to exponential separation in query complexity).

\section{Related works} \label{sec:related-work}

\paragraph{Quantum PAC learning:}

It is instructive to relate our main result on achieving average-case prediction error (Theorem~\ref{thm:noadvquantum}) to quantum probably approximately correct (PAC) learning \cite{servedio2004equivalences, zhang2010improved, arunachalam2016optimal, arunachalam2017guest, chung2019sample, arunachalam2020quantum, caro2020binary}. The latter rigorously established the absence of information-theoretic quantum advantage for learning classical Boolean functions $h: \{0, 1\}^n \rightarrow \{0, 1\}$.
This is a special case of predicting functions of the form $f(x)=\Tr \left( O \cE(\ketbra{x}{x})\right)$.
To see this, we reversibly encode every $n$-bit Boolean function $h$ in a (unitary) CPTP map that acts on $(n+1)$ qubits:
\begin{equation}
    \cE_h(\rho) = U_h \rho U_h^\dagger, \quad \text{where} \quad  U_h\ket{x, a} = \ket{x, h(x) \oplus a} \quad \text{for all} \quad  x \in \{0, 1\}^n, a \in \{0, 1\},
\end{equation}
where $\oplus$ is addition in $\mathbb{Z}_2$.
This is the quantum oracle for implementing the Boolean function $h$.
Fix $O=\mathbb{I}^{\otimes n} \otimes Z=Z_{n+1}$ --- a Pauli-$Z$ operator acting on the final qubit --- to conclude
 $\Tr(O \cE_h(\ketbra{x}{x})) = h(x)$.
In quantum PAC learning, the quantum ML models (or quantum learners) are restricted to quantum samples of the form $\sum_x \sqrt{\mathcal{D}(x)} \ket{x} \ket{h(x)}$, where $\mathcal{D}(x)$ is the probability for sampling $x$ in the input distribution $\mathcal{D}$ and $h(x)$ is the Boolean function in question.
Furthermore, quantum PAC learning focuses on the worst-case input distribution $\mathcal{D}_\star$.
This leaves open the potential for large quantum advantages in more general settings. For instance, the quantum ML could be allowed to access the CPTP map $\cE_h$ or we may want to consider an input distribution with the largest classical-quantum separation.
Theorem~\ref{thm:noadvquantum} closes these open questions showing that no large quantum advantage in sample complexity (or query complexity) is possible even in the much more general setting of learning $f(x) = \Tr(O \cE(\ketbra{x}{x}))$.
Note that direct access to $\cE_h$ would also enable the construction of quantum PAC samples  (if we assume the quantum ML model depends on the fixed input distribution $\cD$). The quantum ML model simply inputs the pure state $\sum_x \sqrt{\mathcal{D}(x)} \ket{x} \ket{0}$ to $\cE_h$ to obtain
\begin{equation}
    U_h \sum_x \sqrt{\mathcal{D}(x)} \ket{x} \ket{0} = \sum_x \sqrt{\mathcal{D}(x)} \ket{x} \ket{h(x)},
\end{equation}
which is the quantum sample considered in quantum PAC learning.
For related works on PAC learning a distribution rather than a function, see \cite{sweke2020quantum, niu2020learnability}.

While existing results \cite{servedio2004equivalences, zhang2010improved, arunachalam2016optimal, arunachalam2017guest, chung2019sample, arunachalam2020quantum, caro2020binary} have shown that no large quantum advantage in sample complexity is possible, we can still have substantial quantum speedups in computational complexity.
In Ref.~\cite{servedio2004equivalences}, for instance, a contrived learning problem is constructed based on factoring. This work showcases the possibility of a quantum advantage in computational complexity, even if no quantum advantage in sample complexity is possible.
\newrevision{
On the other hand, exponential separation in sample complexity is possible if one considers worst-case approximation errors.
}

\paragraph{Quantum machine learning:}

Quantum computers have the potential to improve existing machine learning models based on classical computers.
A series of works require that an exponential amount of data is stored in a quantum random access memory \cite{biamonte2017quantum, harrow2009quantum, rebentrost2014quantum, lloyd2014quantum, lloyd2013quantum, wiebe2012quantum, dunjko2018machine, huang2019near}.
Due to the quantum data structure, one can obtain exponential speed-up for certain machine learning tasks.
However, the act of storing the exponential amount of data will take exponential time. Furthermore, if we assume a similar data structure for classical machines, then the exponential advantage may
vanish altogether \cite{aaronson2015read, tang2019quantum, tang2018quantum, gilyen2018quantum, arrazola2019quantum}.
Another recent line of works focuses on using quantum computers to represent a set of classically intractable functions \cite{farhi2018classification, havlivcek2019supervised, schuld2019quantum, cong2019quantum}.  However, computational power provided by data in a machine learning task can also make a classical ML model stronger \cite{huang2020power} and may challenge some of these proposals in practice.
When we consider learning an arbitrary unitary, due to the inherent complexity of the unitary group, the sample complexity is going to scale with the Hilbert space dimension (exponential in the system size) \cite{poland2020no, sharma2020reformulation}.

\revision{
\paragraph{Classical learning theory:}

For a general introduction to classical learning theory, see the comprehensive book on the foundations of machine learning \cite{mohri2018foundations}.
The classical learning strategy is the same as learning probabilistic real-valued functions.
Existing works on learning real-valued functions \cite{alon1997scale, bartlett1998prediction}, or probabilistic Boolean functions \cite{kearns1994efficient}, mainly focus on the worst-case input distribution
that maximizes sample complexity. Geometric quantities, such as the fat-shattering dimension \cite{alon1997scale}, then provide a characterization of the sample complexity \cite{alon1997scale}.
For example, \cite{bartlett1998prediction} proves that a (worst-case) sample complexity upper bound is given by $\tilde{\mathcal{O}}(\mathrm{fat(\epsilon / 5) / \epsilon^2})$, where $\tilde{\mathcal{O}}(\cdot)$ suppresses logarithmic factors.
In contrast, our proof of Theorem~\ref{thm:noadvquantum} produces a sample complexity upper bound $\mathcal{O}\left(\log(|M^p_{4\epsilon}(\cF_f)|) / \epsilon \right))$ that scales with $1/\epsilon$ instead of $1/\epsilon^2$. The geometric constant $\log(|M^p_{4\epsilon}(\cF_f)|)$ is also different (it measures the cardinality of a maximal packing net that depends on the input distribution).
}

\paragraph{Shadow tomography:}

Shadow tomography is the task of simultaneously estimating the outcome probabilities associated with $M$ 2-outcome measurements up to accuracy $\epsilon$: $p_i (\rho) = \mathrm{tr}(E_i \rho)$, where each $E_i$ is a positive semi-definite matrix with operator norm at most one \cite{aaronson2018shadow,brandao2019sdp,aaronson2019gentle,badescu2020improved}.
The best existing result is given by \cite{badescu2020improved}. They showed that
$
N= \mathcal{O} \left(\log (M)^2 \log (d) /\epsilon^4 \right)
$ copies of the unknown state suffice to achieve this task.
Their protocol is based on an improved quantum threshold search: finding an observable $E_i$ with the expectation value $\Tr(E_i \rho)$ exceeding a certain threshold. Then, they combine this with online learning of quantum states following Aaronson's original protocol. \cite{aaronson2018shadow}.
A more experimentally friendly shadow tomography protocol has been proposed in \cite{huang2020predicting}. It only yields competitive scaling complexities for a restricted set of observables, but can be
implemented on state-of-the-art quantum platforms \cite{struchalin2020experimental,elben2020mixed}.

\revision{\paragraph{PAC learning quantum systems:}

A precursor to shadow tomography is PAC learning of quantum states \cite{aaronson2007learnability, rocchetto2019experimental}, where the goal is to accurately predict expectation values of different observables in an unknown quantum system $\rho$ up to a small average error.
This fits nicely into the scope of Theorem~\ref{thm:noadvquantum}: The optimal fully quantum ML model that can perform quantum data analysis on many copies stored in the quantum memory will not yield a large advantage in sample complexity over ML models that make predictions based solely on classical measurement data from randomized measurements on single copy of $\rho$.
This separation between learning from classical measurement data and learning from the quantum states coherently has not been discussed in \cite{aaronson2007learnability}.
The result \cite{aaronson2007learnability} could be seen as establishing an upper bound on the maximal packing net for the set of CPTP maps $\cF$. This then translates into a sample complexity upper bound needed to achieve good prediction performance.}


\paragraph{Measuring expectation values of Pauli observables:}

Due to the importance of measuring Pauli observables in near-term applications of quantum computers,
a series of methods \cite{crawford2020efficient, izmaylov2019unitary, verteletskyi2020measurement, hamamura2020efficient, jiang2020optimal, huggins2019efficient, bonet2020nearly} have been proposed to reduce the number of measurements needed to estimate
Pauli expectation values.
All of them are based on one basic, yet powerful, observation: Commuting observables can be measured simultaneously.
For example, quantum chemistry applications are often contingent on measuring
 $\mathcal{O}(n^4)$ Pauli observables \cite{peruzzo2014variational}.
The aforementioned Pauli estimation protocols group these observables into $\mathcal{O}(n^3)$ or even only $\mathcal{O}(n)$ commuting groups.
In turn,  $\mathcal{O}(n^3)$ or $\mathcal{O}(n)$ copies of the underlying state suffice to  obtain expectation values for all  $\mathcal{O}(n^4)$ relevant Pauli observables by exploiting the ability to simultaneously measure commuting observables.
On the other hand, the novel technique proposed in Appendix~\ref{app:exp-sep-Pauli}
gets by with even fewer state preparations. A total of $\mathcal{O}(\log(n^4)) = \mathcal{O}(\log(n))$ copies suffice.
Restriction to few-body Pauli observables can yield additional improvements. Several protocols are known for this special case, see e.g.\
 \cite{paini2019approximate, cotler2019quantum, bonet2020nearly, jiang2020optimal, evans2019scalable, huang2020predicting}. 

\revision{\paragraph{Incoherent versus coherent measurements:}

The exponential advantage of quantum ML over classical ML established by our work may be viewed as an exponential separation between coherent measurements (in which a measurement apparatus interacts coherently multiple times with a measured system, storing quantum data which is then processed by a quantum computer) and
incoherent measurements (in which a POVM measurement is performed and the outcome recorded after each interaction between system and apparatus, and the classical measurement outcomes are then processed by a classical computer).
Existing work has shown an
advantage of coherent measurements over \emph{independent} incoherent measurements, a special case in which the POVM measurements do not depend on the results of previous measurements; see e.g., \cite{haah2017sample} on quantum state tomography and \cite{huang2020predicting} on shadow tomography.
However, few prior results limit the power of incoherent measurements in the \emph{adaptive} setting, in which each measurement performed may depend on the outcomes of all the previous measurements.

A prior result we were aware of before obtaining our result was \cite{bubeck2020entanglement}, showing a mild polynomial advantage of coherent over incoherent measurements for the task of distinguishing whether an unknown quantum state is close to the completely mixed state or not.
Our work established an exponential separation between incoherent and coherent measurements in sample complexity (the number of copies of a quantum state needed to perform the task) for shadow tomography.
After our work was complete, \cite{aharonov2021quantum} also presented a detailed analysis of the power of coherent and incoherent measurements, finding
an exponential separation between incoherent and coherent measurements for the task of distinguishing between different types of quantum channels.}

\section{Details of numerical experiments} \label{sec:numerics}

\subsection{Mixed states}

For the case of mixed states given by $\rho = (I + P_{x^*}) / 2^n$, we have $f(x) = \Tr(P_x \rho)= \delta_{x,x^*}$ for all $x \in \{I, X, Y, Z\}^n$. That is, $f (x)$ is a point function, i.e.\ $f(x^*) =1$ for exactly one $x^*$, all other Pauli strings evaluate to zero.

\paragraph{Restricted classical ML:}

We consider a restricted classical ML model that implements an exhaustive search over all Pauli observables $P_x$ with $x \in \left\{I,X,Y,Z\right\}^n$.
For each $x$, we repeatedly measure the observable $P_x$ and check whether the outcome $-1$  ever occurs. If this is the case, we know $f(x)=0$ with certainty (recall, that $f(x) \in \left\{0,1\right\}$ is a point function). After looping through all observables, we will be left with a single input $x^*$ that obeys $f(x^*)=1$. There are a total of $4^n$ inputs and whenever $x \neq x^*$, the expected number of measurements required to obtain the outcome $-1$ is $1/(\mathrm{Pr}_x[-1])=2$. In contrast, for $x=x^*$, outcome $-1$ can never occur. This results in a sample complexity of $\mathcal{O}(4^n)$ -- a scaling that is confirmed by our numerical experiments and matches the lower bound $\Omega(4^n)$.

\paragraph{Classical ML:}

For classical ML, we implement property prediction with classical shadows based on random Clifford measurements~
\cite{huang2020predicting}. The sample complexity to
predict $M$ Pauli observables up to small constant accuracy is known to be $\mathcal{O}(\max_{x} \Tr(P_x^2) \log(M))$. Using $\Tr (P_x^2)=\Tr (I^{\otimes n})=2^n$ and $M=4^n$, this upper bound simplifies to $\mathcal{O}(n 2^n)$. The numerical experiments confirm this theoretical prediction.

\paragraph{Quantum ML:}

For quantum ML, we use the procedure introduced in Appendix~\ref{sec:pauli-qml}.
We first perform several repetitions of two-copy Bell basis measurements to estimate absolute values $|\Tr(P_x \rho)|^2$ for all $4^n$ possible inputs $x$.
This allows us to immediately identify $x^*$.
One could solve for $x^*$ by performing Gaussian elimination in $GF(2)$. A similar strategy has also been used for learning quantum channels \cite{harper2020fast}.
The required sample complexity is $\mathcal{O}(\log(4^n)) = \mathcal{O}(n)$ and the numerical experiments confirm this linear scaling.

\subsection{Product states}

We consider the unknown $n$-qubit state to be a tensor product of $n$ single-qubit stabilizer states (there are only six choices $\ket{0}, \ket{1}, \ket{+}, \ket{-}, \ket{y, +}, \ket{y, -}$).
We only need to obtain measurement outcomes for single-qubit Pauli observables to completely determine such product states.

\paragraph{Restricted classical ML:}
Recall that the goal of this task is to predict $f(x) = \Tr(P_x \rho)$ for an unknown quantum state $\rho$.
The restricted classical ML can collect measurement data of the form $\{(x_i, o_i)\}$, where $x_i \in \{I, X, Y, Z\}^n$ and $o_i \in \pm 1$ is the measurement outcome when we measure $P_x$ on $\rho$.
We simply collect measurement outcomes for the $3n$ single-qubit Pauli observables by choosing the appropriate $x_i$. For each qubit, we
sample $\pm 1$-outcomes in the $X$-, $Y$- and $Z$-basis. Once we find that two of the three Pauli observables result in both the outcome $\pm 1$, we can determine the single-qubit state for the particular qubit. For two of the three bases, the underlying distribution is uniform, while deterministic outcomes are produced in the third basis.
In turn, we expect to require $6n$ Pauli measurements to unambiguously characterize the underlying stabilizer state. This scaling is confirmed by the numerical experiments.

\paragraph{Classical ML:}
We associate classical ML models with reconstruction procedures that can perform arbitrary POVM measurements on individual copies of the unknown state $\rho$.
Here, we consider a sequence of POVMs where we measure first in the all-$X$ basis, second in the all-$Y$ basis, and then in the all-$Z$ basis (and repeat).
Similar to the restricted classical ML, we also check if two of the three possible single-qubit observables $X, Y, Z$ have resulted in both outcomes $\pm 1$.
Once two of the three single-qubit observables have produce both outcomes $\pm 1$, we can perfectly identify the corresponding single-qubit stabilizer state.
But, to complete the assignment, we need to be able to do so for all $n$ qubits. This incurs an additional logarithmic factor in the expected number of measurements. We expect to require $\mathcal{O}(\log(n))$ measurement repetitions to unambiguously determine the underlying product state.
Numerical experiments confirm this scaling behavior.

\paragraph{Quantum ML:}
Quantum ML models can perform quantum data processing on multiple copies of quantum states. For product states, we only perform the following procedure.
For each repetition, we perform two-copy Bell basis measurements on $\rho \otimes \rho$ to simultaneously measure $X\otimes X, Y \otimes Y, Z \otimes Z$ on each of the $n$ qubits in $\rho$.
We can determine the eigenbasis (X-basis, Y-basis, or Z-basis) of each qubit when we found that two of the observables $X\otimes X, Y \otimes Y, Z \otimes Z$ result in both $\pm 1$ outcomes.
For two of the three observables, the underlying distribution is uniform in $\pm 1$, while deterministic outcomes are produced in the third observable.
Since we know the stabilizer bases for each qubit after a few two-copy measurements, we simply measure the $n$ qubits in the corresponding stabilizer basis on the final copy to recover the full state.
We refer to \cite{gross2017schur, montanaro2017learning,zhao2016fast,montanaro2016survey} for results on efficient procedures for testing and learning stabilizer states in general.

\section{Proof of Theorem~\ref{thm:noadvquantum}}\label{sec:proofnoadv}

This section contains a thorough treatment of \emph{average} prediction errors. We consider related setups for the classical and quantum learning settings.

The learning problem is defined by a set of CPTP maps $\cF$, an input distribution $\mathcal{D}$, and an observable $O$ with $\norm{O} \leq 1$.
Each CPTP map $\cE \in \cF$ maps a $n$-qubit quantum state to $m$-qubit state. This collection defines a function
\begin{equation}
    f_{\cE}(x) = \Tr(O \cE(\lvert x \rangle \! \langle x \rvert)): \{0, 1\}^n \rightarrow \mathbb{R}
\end{equation}
The goal is to learn a function $f: \{0, 1\}^n \rightarrow \mathbb{R}$
such that with high probability
\begin{equation}
    \mathbb{E}_{x \sim \mathcal{D}} |f(x) - f_{\cE}(x)|^2 \quad \text{is small}.
\end{equation}
A bit of additional context is appropriate here:
we are studying the existence of learning algorithms under a fixed learning problem defined by input distribution $\cD$, observable $O,$ and set of CPTP maps $\cF$.
In turn, the actual learning algorithms may and, in general, will depend on these 
mathematical objects.

One of our main technical contributions -- Theorem~\ref{thm:noadvquantum} -- highlights that a substantial quantum advantage is impossible for this setting (small \emph{average-case} prediction error). This is in stark contrast to the setting of achieving small \emph{worst-case} prediction error.
The proof consists of two parts.
Section~\ref{sec:lowerbdQU} establishes a lower bound for the query complexity of \emph{any} quantum ML model.
Subsequently, Section~\ref{sec:uppbdCL} provides an upper bound for the query complexity achieved by \emph{certain} classical ML models.
Finally, a combination of these two results establishes Theorem~\ref{thm:noadvquantum}, see  Section~\ref{sec:combineCLQU}.

\subsection{Information-theoretic lower bound for quantum machine learning models} \label{sec:lowerbdQU}

The quantum machine learning model consists of learning phase and prediction phase. In the learning phase, the quantum ML model accesses the quantum experiment characterized by the CPTP map $\cE$ for $N_{\mathrm{Q}}$ times to learn a model.
We consider the quantum ML model to be a mixed state quantum computation algorithm (a generalization of unitary quantum computation). Starting point is an initial state $\rho_0$ on any number of qubits. Subsequently, arbitrary quantum operations $\mathcal{C}_t$ (CPTP maps) are interleaved with in total $N_{\mathrm{Q}}$ invocations $\mathcal{E} \otimes \mathcal{I}$ of the unknown (black box) CPTP map and produce a final state
\revision{
\begin{equation} \label{eq:outputQSQML}
    \rho_{\cE} = \mathcal{C}_{N_{\mathrm{Q}}} (\cE \otimes \mathcal{I}) \mathcal{C}_{N_{\mathrm{Q}}-1} \ldots \mathcal{C}_1 (\cE \otimes \mathcal{I}) (\rho_0).
\end{equation}
}
In this model, we can assume without loss that $\mathcal{E}$ always acts on the first $n$ qubits, because the quantum operations $\mathcal{C}_t$ are unrestricted. In particular, they could contain certain SWAP operations that permute the qubits around.
The final state $\rho_{\cE}$ is the quantum memory that stores the prediction model learned from the CPTP map $\cE$ using the quantum ML algorithm. Obtaining $\rho_{\cE}$ concludes the quantum learning phase.

In the prediction phase,
we assume that new inputs are provided as part of a sequence
\begin{equation}
    x_1, x_2, x_3, \ldots \in \{0, 1\}^n.
\end{equation}
For each sequence member $x_i$, the quantum ML model accesses the input $x_i$, as well as the current quantum memory. 
It produces an outcome by performing a POVM measurement on the quantum memory $\rho_{\cE}$. We emphasize that this can, and in general will, affect the quantum memory nontrivially. The quantum ML outputs $h_{\mathrm{Q}}(x_i)$ depend on the entire sequence $x_1,\ldots,x_i$. And different sequence orderings will produce different predictions.
For example, when $n=2$, the following ordering may result in the prediction
\begin{equation}
    x_1 = 00, x_2 = 01, x_3 = 10, x_4 = 11 \,\, \rightarrow \,\, h_{\mathrm{Q}}(00) = -0.3, h_{\mathrm{Q}}(01) = 0.5, h_{\mathrm{Q}}(10) = 0.2, h_{\mathrm{Q}}(11) = -0.7,
\end{equation}
but a different ordering may result in a slightly different prediction, such as
\begin{equation}
    x_1 = 11, x_2 = 01, x_3 = 10, x_4 = 00 \,\, \rightarrow  \,\, h_{\mathrm{Q}}(00) = -0.2, h_{\mathrm{Q}}(01) = 0.5, h_{\mathrm{Q}}(10) = 0.2, h_{\mathrm{Q}}(11) = -0.6.
\end{equation}
Also, note that $h_{\mathrm{Q}}(x_i)$ can be randomized because a quantum measurement is performed to produce the prediction outcome.
The ordering does not affect the theorem we want to prove. In the following, we will fix the input ordering to be an arbitrary ordering. For example, we can use the input ordering such that the quantum ML model has the smallest prediction error.

After fixing an input ordering, we can treat the entire prediction phase (taking a sequence of inputs $x_1, x_2, \ldots$ and producing $h_{\mathrm{Q}}(x_1), h_{\mathrm{Q}}(x_2), \ldots$) as an enormous POVM measurement on the output state $\rho_{\cE}$ obtained from the learning phase.
Each outcome $a$ from the enormous POVM measurement on the output state $\rho_{\cE}$ corresponds to a function $h_{\mathrm{Q}, a}(x): \{0, 1\}^n \rightarrow \mathbb{R}$.
Using Naimark's dilation theorem, every POVM measurement is a projective measurements on a larger Hilbert space. Since the quantum memory that the quantum ML model can operate on contains an arbitrary amount of qubits, we can use Naimark's dilation theorem to restrict the enormous POVM measurement to a projective measurement $\{P_a\}_a$.
Hence, for any CPTP map $\cE \in \cF$, when we asks the quantum ML model to produce the prediction for an ordering of inputs $x_1, x_2, \ldots$, the output values $h_{\mathrm{Q}}(x_1), h_{\mathrm{Q}}(x_2), \ldots$ will be given by
\begin{equation} \label{eq:h-QaofQML}
    h_{\mathrm{Q}, a}(x) \quad \text{with probability} \quad \Tr(P_a \rho_{\cE}) = \Tr(P_a \mathcal{C}_{N_{\mathrm{Q}}} (\cE \otimes\, \mathcal{I}) \mathcal{C}_{N_{\mathrm{Q}}-1} \ldots \mathcal{C}_2 (\cE \otimes\, \mathcal{I}) \mathcal{C}_1 (\cE \otimes\, \mathcal{I}) \mathcal{C}_0(\rho_0)),
\end{equation}
for a projective measurement $\{P_a\}_a$ with $\sum_a P_a = I$.

Finally, we will assume that the produced function $h_{\mathrm{Q}}(x)$ achieves small prediction error
\begin{equation} \label{eq:QMLassumption}
    \E_{x \sim \mathcal{D}} \left| h_{\mathrm{Q}}(x) - \Tr(O \mathcal{E}(\lvert x \rangle \! \langle x \rvert)) \right|^2 \leq \epsilon \quad \text{with probability at least $2/3$.}
\end{equation}
for any CPTP map $\mathcal{E} \in \cF$.
This assumption asserts that
\begin{equation} \label{eq:prob2/3}
    \sum_{a=1}^A \Tr(P_a \rho_{\cE}) \mathbbm{1}\left[\E_{x \sim \mathcal{D}} \left| h_{\mathrm{Q}, a}(x) - \Tr(O \mathcal{E}(\lvert x \rangle \! \langle x \rvert)) \right|^2 \leq \epsilon\right] \geq 2/3,
\end{equation}
where $\mathbbm{1}[z]$ denotes the indicator function of event $z$. That is, $\mathbbm{1}[z]=1$ if $z$ is true and $\mathbbm{1}[z]=0$ otherwise.

\subsubsection{Maximal packing net}
\label{sec:packnetQ}

We emphasize that Rel.~\eqref{eq:prob2/3} must be valid for \emph{any} $\mathcal{E} \in \mathcal{F}$.
Because we only need to output a function $h_{\mathrm{Q}}(x)$ that approximates $f(x) = \Tr(O \cE(\ketbra{x}{x}))$ on average, the task will not be hard when there are only a few qualitatively different CPTP maps in $\mathcal{F}$.
However, the problem could become harder when $\mathcal{F}$ contains a large amount of very different CPTP maps.
The task is now to transform this requirement into a stringent lower bound on $N_{\mathrm{Q}}$ -- the number of black-box uses of the unknown CPTP map $\mathcal{E} \otimes \mathcal{I}$ within the quantum computation \eqref{eq:outputQSQML}. As a starting point, we equip the set of target functions $\cF_f = \{f_{\cE}(x) = \Tr(O \mathcal{E}(\lvert x \rangle \! \langle x \rvert)) | \cE \in \cF\}$ with a packing  net. Packing nets are discrete subsets whose elements are guaranteed to have a certain minimal pairwise distance (think of spheres that must not overlap with each other). We choose points (functions) $f_{\mathcal{E}_i} \in \mathcal{F}_f$ and demand
\begin{equation} \label{eq:epspackingnetdistance}
    \E_{x \sim \mathcal{D}} |f_{\cE_i}(x) - f_{\cE_j}(x)|^2 > 4\epsilon \quad \text{whenever} \quad i \neq j.
\end{equation}
We denote the resulting packing net of $\mathcal{F}_f$ by $M^p_{4\epsilon}(\cF_f)$ and note that every such set has finitely many elements ($\mathcal{F}_f$ is a compact set).
We also assume that $M^p_{4\epsilon}(\cF_f)$ is maximal in the sense that no other $4\epsilon$-packing net can contain more points (functions).

It is possible to utilize packing nets to derive a query complexity lower bound for the quantum machine learning model.
In fact, we will present two different proof strategies.
The first proof is inspired by \cite{flammia2012quantum,haah2017sample,huang2020predicting} and analyzes a communication protocol.
The
second proof uses a proof technique that depends on an analysis of
polynomials similar to \cite{beals2001quantum}. While it is somewhat weaker than the
information-theoretic bound in the first proof, we include the
derivation for completeness as we believe that it may be insightful
for the interested reader.


\subsubsection{Proof strategy I: mutual information analysis} \label{sec:MIanalysisML}

Let us define a communication protocol between two parties, say Alice and Bob.
They use the packing net $M^p_{4 \epsilon}(\cF_f)$ as a dictionary to communicate randomly selected classical messages. More precisely,
Alice samples an integer $X$ uniformly at random from  $1,2, \ldots, |M^p_{4\epsilon}(\cF_f)|$ and chooses the corresponding CPTP map $\cE_X \in  M^p_{4\epsilon}(\cF_f)$.
When Bob wants to access the unknown CPTP map $\cE_X$, he will ask Alice to apply the CPTP map $\cE_X$.
Bob will then execute the quantum machine learning model \eqref{eq:outputQSQML} to obtain a prediction model $h_{\mathrm{Q}, a}(x)$, where $a$ parameterizes the prediction model.
Subsequently, Bob solves the following optimization problem
\begin{equation}
    \tilde{X} = \argmin_{X' = 1, \ldots, |M^p_{4\epsilon}(\cF_f)|} \E_{x \sim \mathcal{D}} \left| h_{\mathrm{Q}, a}(x) - \Tr(O \mathcal{E}_{X'}(\lvert x \rangle \! \langle x \rvert)) \right|^2
   \label{eq:bob-decoding}
\end{equation}
to obtain an integer $\tilde{X}$.
This decoding procedure seems adequate, provided that the prediction model $h_{\mathrm{Q}}$ approximately reproduces the true underlying function.
More precisely, assumption~\eqref{eq:QMLassumption} asserts
\begin{equation}
    \E_{x \sim \mathcal{D}} \left| h_{\mathrm{Q}, a}(x) - \Tr(O \mathcal{E}_X(\lvert x \rangle \! \langle x \rvert)) \right|^2 \leq \epsilon \quad \text{with probability at least $2/3$.}
\end{equation}
Here is where the choice of dictionary matters: $M^p_{4\epsilon}(\cF_f)$ is a packing net, see Equation~\eqref{eq:epspackingnetdistance}.
For $X' \neq X$ this necessarily implies
\begin{align}
    \E_{X \sim \mathcal{D}} \left| h_{\mathrm{Q}, a}(x) - f_{\cE_{X'}}(x) \right|^2 &\geq \left( \Big(\E_{X \sim \mathcal{D}} \left| f_{\cE_{X}}(x) - f_{\cE_{X'}}(x) \right|^2\Big)^{1/2} - \Big(\E_{X \sim \mathcal{D}} \left| h_{\mathrm{Q}, a}(x) - f_{\cE_{X}}(x) \right|^2\Big)^{1/2}\right)^2 \\
    &> \left( 2 \sqrt{\epsilon} - \sqrt{\epsilon} \right)^2 = \epsilon.
\end{align}
This allows us to conclude that Bob's decoding strategy \eqref{eq:bob-decoding}
succeeds perfectly if
\begin{equation}
    \E_{x \sim \mathcal{D}} \left| h_{\mathrm{Q}, a}(x) - \Tr(O \mathcal{E}_X(\lvert x \rangle \! \langle x \rvert)) \right|^2 \leq \epsilon.
\end{equation}
In turn, Assumption~\ref{eq:QMLassumption} ensures $\tilde{X}=X$ (perfect decoding) with probability at least $2/3$.

Now, we use the fact that Alice samples her message $X$ uniformly at random from a total of $|M^p_{4\epsilon}(\cF_f)|$ integers.
\revision{Because $\tilde{X}=X$ (perfect decoding) with probability at least $2/3$, Fano's inequality implies that
\begin{equation}
    H(X | \tilde{X}) \leq H(1/3) + \log(|M^p_{4\epsilon}(\cF_f)|) / 3,
\end{equation}
where $H(x) = -x \log x - (1-x) \log (1-x)$ is the binary entropy.
This gives a lower bound on the mutual information between sent and decoded message, namely
\begin{equation} \label{eq:MutualInfo-Mp}
    I(X : \tilde{X}) = H(X) - H(X | \tilde{X}) \geq \frac{2}{3} \log(|M^p_{4\epsilon}(\cF_f)|) - H(1/3) =  \Omega\left(\log(|M^p_{4\epsilon}(\cF_f)|)\right).
\end{equation}}
Next, note that $\tilde{X}$ is obtained by classically processing a measurement outcome $a$ of the quantum state $\rho_{\mathcal{E}_X}$ The data processing inequality and Holevo's theorem \cite{holevo1973bounds, horodecki2009quantum,bengtsson2017geometry,araki2002entropy} then imply
\begin{equation} \label{eq:MutualInfo-DP}
    I(X : \tilde{X}) \leq I(X : a) \leq \chi(X : \rho_{\cE_X}).
\end{equation}
The Holevo $\chi$ quantity between the classical random variable $X$ and the quantum state $\rho_{\cE_X}$ is
\begin{equation}
    \chi(X : \rho_{\cE_X}) = S\left( \E_X \rho_{\cE_X} \right) - \E_X S\left(\rho_{\cE_X}\right),
\end{equation}
where $S(\rho) = \Tr(-\rho \log \rho)$ is the von Neumann entropy. Throughout this work, we refer to $\log$ with base $\mathrm{e}$. Recall that Bob produces $\rho_{\cE_X}$ by utilizing a total of $N_{\mathrm{Q}}$ channel copies obtained from Alice.
We can use the specific layout~\eqref{eq:outputQSQML} of Bob's quantum computation to produce an upper bound on the Holevo-$\chi$:
\begin{equation}
\chi(X : \rho_{\cE_X}) \leq \mathcal{O}(m N_{\mathrm{Q}})
\label{eq:holevo-upper-bound}
\end{equation}
This bound follows from induction over a sample-resolved variant of Bob's quantum computation. For $t =0,1,\ldots,N_{\mathrm{Q}}$, we will show that
\begin{equation}
 \rho^t_{\cE} = \mathcal{C}_{t} (\cE \otimes\, \mathcal{I}) \mathcal{C}_{t-1} \ldots \mathcal{C}_1 (\cE \otimes\, \mathcal{I}) \mathcal{C}_0(\rho_0)
 \quad \text{obeys} \quad \chi(X : \rho^t_{\cE_X}) \leq (2\log 2) m t.
 \end{equation}
 Bound~\eqref{eq:holevo-upper-bound} then follows from recognizing that setting $t=N_{\mathrm{Q}}$ reproduces Bob's complete computation, see Equation~\eqref{eq:outputQSQML}.

 The base case ($t=0$) is simple, because $\rho^0_{\cE_X} = \mathcal{C}_0(\rho_0)$ does not depend on $X$ at all. This ensures
 \begin{equation}
 \chi \left(X: \rho_{\cE_X}^0 \right) = S\left( \E_X \rho_{\cE_X}^0 \right) - \E_X S\left(\rho_{\cE_X}^0\right)= S \left( \rho_{\cE_X}^0 \right) - S\left(\rho_{\cE_X}^0\right) = 0
 \leq (2 \log 2) m t \quad (t=0).
 \end{equation}
Now, let us move to the induction step ($t>0$). The induction hypothesis provides us with
\begin{equation} \label{eq:induchypoIX}
    \chi(X : \rho^{t-1}_{\cE_X}) \leq (2\log 2)m (t-1)
\end{equation}
and we must  relate $\chi(X : \rho^{t}_{\cE_X})$ to $\chi(X : \rho^{t-1}_{\cE_X})$.
To achieve this goal, we use the fact that the Holevo-$\chi$ is closely related to the quantum relative entropy $D(\rho || \sigma) = \Tr\left( \rho (\log \rho - \log \sigma) \right)$ \cite{horodecki2009quantum,bengtsson2017geometry,araki2002entropy}. Indeed,
\begin{equation}
\chi(X : \rho^t_{\cE_X})
= \E_X\left[ \Tr\left( \rho^t_{\cE_X} \log \rho^t_{cE_X} - \rho^t_{\cE_X} \log \left(\E_{X'}\rho^t_{cE_{X'}} \right) \right) \right]
=  \E_{X} D\left(\rho^t_{\cE_X} || \E_{X'}\rho^t_{cE_{X'}}\right),
\end{equation}
and monotonicity of the quantum relative entropy asserts
\begin{align*}
\E_{X} D\left(\rho^t_{\cE_X} || \E_{X'}\rho^t_{cE_{X'}}\right)
=&  \E_{X} D\left( \mathcal{C}_{t}\left( \left(\cE_{X} \otimes\, \mathcal{I}\right) \left( \rho^{t-1}_{\cE_X} \right) \right) \,\, || \,\, \mathcal{C}_{t}\left( \E_{X'} \left(\cE_{X'} \otimes\, \mathcal{I}\right)\left( \rho^{t-1}_{\cE_{X'}} \right)\right)\right)\\
    \leq & \E_{X} D\left( \left(\cE_{X} \otimes\, \mathcal{I}\right) \left( \rho^{t-1}_{\cE_X} \right) \,\, || \,\, \E_{X'} \left(\cE_{X'} \otimes\, \mathcal{I}\right)\left( \rho^{t-1}_{\cE_{X'}} \right)\right)\\
    =& S\left( \E_X \left(\cE_{X} \otimes\, \mathcal{I}\right)(\rho^{t-1}_{\cE_X}) \right) - \E_X S\left( \left(\cE_{X} \otimes\, \mathcal{I}\right)(\rho^{t-1}_{\cE_X})\right).
\end{align*}
This effectively allows us to ignore the $t$-th quantum operation $\mathcal{C}_t$ and instead exposes the $t$-th invocation of $\mathcal{E} \otimes \mathcal{I}$.

We analyze the two remaining terms separately.
Let use define the notation $\Tr_{\leq m}$ as the partial trace over the first $m$ qubits, and $\Tr_{> m}$ as the partial trace over the rest of the qubits.
Subadditivity of the von Neumann entropy $S(\rho)$ \cite{horodecki2009quantum,bengtsson2017geometry,araki2002entropy} implies
\begin{align}
    S\left( \E_X \left(\cE_{X} \otimes\, \mathcal{I}\right)(\rho^{t-1}_{\cE_X}) \right) &\leq S\left( \Tr_{\leq m} \E_X \left(\cE_{X} \otimes\, \mathcal{I}\right)(\rho^{t-1}_{\cE_X}) \right) + S\left( \Tr_{> m} \E_X \left(\cE_{X} \otimes\, \mathcal{I}\right)(\rho^{t-1}_{\cE_X}) \right),\\
    &\leq S\left( \Tr_{\leq m} \E_X \left(\cE_{X} \otimes\, \mathcal{I}\right)(\rho^{t-1}_{\cE_X}) \right) + m \log 2\\
    &= S\left( \Tr_{\leq n} \E_X \rho^{t-1}_{\cE_X} \right) + m \log 2.
\end{align}
The second inequality uses the fact that the maximum entropy for an $m$-qubit system is at most $m \log 2$. The last equality is due to the following technical observation (the action of a CPTP map can be traced out).

\begin{lemma}
Fix a CPTP map $\mathcal{E}$ from $n$ qubits to $m$ qubits and let $\mathcal{I}$ denote the identity map on $n' \geq 0$ qubits. Then, $\Tr_{\leq m}[(\cE \otimes \mathcal{I}) \rho] = \Tr_{\leq n}[\rho]$ for any $(n+n')$-qubit state $\rho$.
\end{lemma}
\begin{proof}
Let $\mathcal{E}(\rho) = \sum_i K_i \rho K_i^\dagger$ be a Kraus representation of the CP map $\mathcal{E}$. TP moreover implies $\sum_i K_i^\dagger K_i = I$. For any input state $\rho$,
Linearity and (partial) cyclicity of the partial trace then ensure
\begin{align*}
\Tr_{\leq m}\left((\cE \otimes \mathcal{I}) \rho\right) =&
\sum_i \Tr_{\leq m} \left( K_i \otimes I\rho K_i^\dagger \otimes I\right)
= \sum_i \Tr_{\leq m} \left( \rho (K_i^\dagger K_i) \otimes I \right)
= \Tr_{\leq m} \left(\rho I \otimes I \right)
= \Tr_{\leq m} \left( \rho \right).
\end{align*}
This concludes the proof of the lemma.
\end{proof}

Similarly, the second term can be lower bounded by
\begin{align}
    \E_X S\left( \left(\cE_{X} \otimes\, \mathcal{I}\right)(\rho^{t-1}_{\cE_X}) \right) &\geq \E_X S\left( \Tr_{\leq m} \left(\cE_{X} \otimes\, \mathcal{I}\right)(\rho^{t-1}_{\cE_X}) \right) - \E_X S\left( \Tr_{> m} \left(\cE_{X} \otimes\, \mathcal{I}\right)(\rho^{t-1}_{\cE_X}) \right),\\
    &\geq \E_X  S\left( \Tr_{\leq m} \left(\cE_{X} \otimes\, \mathcal{I}\right)(\rho^{t-1}_{\cE_X}) \right) - m \log 2,\\
    &= \E_X S\left( \Tr_{\leq n} \rho^{t-1}_{\cE_X} \right) - m \log 2.
\end{align}
We can combine these two bounds with the monotonicity of the quantum relative entropy to obtain
\begin{align}
    \chi(X : \rho^t_{\cE_X}) &\leq S\left( \E_X \left(\cE_{X} \otimes\, \mathcal{I}\right)(\rho^{t-1}_{\cE_X}) \right) - \E_X S\left( \left(\cE_{X} \otimes\, \mathcal{I}\right)(\rho^{t-1}_{\cE_X})\right)\\
    &\leq S\left( \Tr_{\leq n} \E_X \rho^{t-1}_{\cE_X} \right) - \E_X S\left( \Tr_{\leq n} \rho^{t-1}_{\cE_X} \right) + (2\log 2)m \\
    &= \E_X D\left( \Tr_{\leq n} \rho^{t-1}_{\cE_X} ||  \Tr_{\leq n} \E_{X'} \rho^{t-1}_{\cE_{X'}}\right) + (2\log 2)m\\
    &\leq \E_X D\left( \rho^{t-1}_{\cE_X} || \E_{X'} \rho^{t-1}_{\cE_{X'}} \right) + (2 \log 2) m\\
    & = \chi(X : \rho^{t-1}_{\cE_X}) + (2 \log 2) m.
\end{align}
Plug in the induction hypothesis \eqref{eq:induchypoIX} to complete the argument:
\begin{equation}
\chi(X : \rho^t_{\cE_X})
\leq \chi(X : \rho^{t-1}_{\cE_X}) + (2 \log 2) m \leq (2 \log 2 ) m(t-1) + (2 \log 2) m  = (2 \log 2) mt. \label{eq:holevo-upper-bound2}
\end{equation}

It is worthwhile to pause and recapitulate the main insights from this section: i.) a lower bound on the mutual information in terms of packing net cardinality, see Equation~\eqref{eq:MutualInfo-Mp};  ii.) Holevo's theorem, see Equation~\eqref{eq:MutualInfo-DP}; and, (iii) an upper bound on the Holevo-$\chi$ in terms of query complexity, see Equation~\eqref{eq:holevo-upper-bound2} for $t=N_{\mathrm{Q}}$.
Combining all of them yields
\begin{equation*}
\Omega\left(\log(|M^p_{4\epsilon}(\cF_f)|)\right) \leq I \left( X: \tilde{X}\right)
\leq \chi \left( X: \rho_{\cE_X} \right) \leq (2 \log 2) mN_{\mathrm{Q}}.
\end{equation*}
Rearranging this display yields a lower bound on the minimal query complexity in terms of packing net size:
\begin{equation} \label{eq:NQlowerbound}
    N_{\mathrm{Q}} \geq \Omega\left(\frac{\log(|M^p_{4\epsilon}(\cF_f)|)}{m}\right).
\end{equation}

\subsubsection{Proof strategy II: polynomial method}


\revision{The second proof uses a proof technique that depends on analysis of polynomials~\cite{beals2001quantum}.
It leads to somewhat weaker results that only apply if $m \leq n$. We include this derivation for completeness as we believe that it may be insightful for the interested reader.}

Let us start by recalling that we may embed a $4\epsilon$-packing net $M^p_{4\epsilon}(\cF_f)$ within the set of target functions $\cF_f$. Geometrically, this means that each $\cE \in M^p_{4\epsilon}(\cF_f)$ describes the center of a $2\epsilon$-ball (this radius is defined with respect to average prediction error squared). And, according to the defining property Equation~\eqref{eq:epspackingnetdistance}, these balls do not overlap. We can use these disjoint balls to cluster different quantum machine learning solutions. Define
\begin{equation} \label{eq:smallballFQ}
    \mathcal{F}^{\mathrm{Q}}_{\cE} = \left\{ a \in A \,\, : \,\, \E_{x \sim \mathcal{D}} \left| h_{\mathrm{Q}, a}(x) - f_{\cE}(x) \right|^2 \leq \epsilon \right\},
\end{equation}
where $A$ is a placeholder for all possible answers the quantum machine learning model can provide.
See the definition given in Equation~\eqref{eq:h-QaofQML}.
The packing net condition~\eqref{eq:epspackingnetdistance} ensures that different clusters are completely disjoint. For distinct $\cE_1,\cE_2 \in \cF$ and $a_1 \in \cF^{\mathrm{Q}}_{\cE_1}$, $a_2 \in \cF^{\mathrm{Q}}_{\cE_2}$, two triangle inequalities and Equation~\eqref{eq:epspackingnetdistance} yield
\begin{align}
    \sqrt{\E_{x \sim \mathcal{D}} \left| h_{\mathrm{Q}, a_1}(x) - h_{\mathrm{Q}, a_2}(x) \right|^2}
    &\geq \sqrt{\E_{x \sim \mathcal{D}} \left| f_{\cE_1}(x) - h_{\mathrm{Q}, a_2}(x) \right|^2} - \sqrt{\E_{x \sim \mathcal{D}} \left| h_{\mathrm{Q}, a_1}(x) - f_{\cE_1}(x) \right|^2} \\
    &\geq \sqrt{\E_{x \sim \mathcal{D}} \left| f_{\cE_1}(x) - f_{\cE_2}(x) \right|^2} - \sqrt{\E_{x \sim \mathcal{D}} \left| h_{\mathrm{Q}, a_2}(x) - f_{\cE_2}(x) \right|^2} \\
    &- \sqrt{\E_{x \sim \mathcal{D}} \left| h_{\mathrm{Q}, {a_1}}(x) - f_{\cE_1}(x) \right|^2} \\
    &> 2 \sqrt{\epsilon} - \sqrt{\epsilon} - \sqrt{\epsilon} = 0.
\end{align}
This implies $f_{a_1}^Q \neq f_{a_2}^Q$ and, more importantly, $\mathcal{F}^Q_{\cE_1} \cap \mathcal{F}^Q_{\cE_2} = \varnothing$ whenever $ f_{\cE_1} \neq f_{\cE_2}$.

We will use this insight to reason about an auxiliar matrix $P$ of size
$|M^p_{4\epsilon}(\cF_f)| \times |M^p_{4\epsilon}(\cF_f)|$.
We label rows and columns by packing net elements $f_{\cE_i}$ with $i=1$ (rows) or $i=2$ (columns). For each pair $f_{\cE_1},f_{\cE_2}$, let $P_{\cE_1,\cE_2}$ denote the probability of a mix-up between $\cE_1$ and $\cE_2$. Such mix-ups occur if the underlying CPTP map is $\cE_2$, but the quantum ML model outputs an answer $a \in \cF^{\mathrm{Q}}_{\cE_1}$ that belongs to the cluster associated with $\cE_1$:
\begin{align}
    P_{\cE_1, \cE_2} &= \sum_{a=1}^A \Tr(P_a \rho_{\cE_2}) \mathbbm{1}\left[\E_{x \sim \mathcal{D}} \left| h_{\mathrm{Q}, a}(x) - f_{\cE_1}(x) \right|^2 \leq \epsilon\right]
    = \sum_{a: h_{\mathrm{Q}, a} \in \mathcal{F}^Q_{\cE_1}} \Tr(P_a \rho_{\cE_2}).
    \label{eq:polynomial-method-P}
\end{align}
Here, $\rho_{\epsilon_2}$ is the outcome state of the quantum ML model (trained on CPTP map $\cE_2$) and $P_a$ is the POVM element associated with predicting $a_1$.
Recall that the main assumption on the quantum ML model is that it predicts accurately with probability at least $2/3$.
This implies
\begin{equation}
    P_{\cE_1, \cE_1} \geq 2 / 3
    \quad \text{for each}  \quad \cE_1 \in M^p_{4\epsilon}(\cF_f),
\end{equation}
while each row sum over off-diagonal matrix elements is strictly smaller. For $f_{\cE_2} \in M^p_{4\epsilon}(\cF_f)$,
\begin{align}
    \sum_{\substack{f_{\cE_1} \in M^p_{4\epsilon}(\cF_f) \\ f_{\cE_1} \neq f_{\cE_2}}} P_{\cE_1, \cE_2}
    &= \sum_{\substack{f_{\cE_1} \in M^p_{4\epsilon}(\cF_f) \\ f_{\cE_1} \neq f_{\cE_2}}} \sum_{a: h_{\mathrm{Q}, a} \in \mathcal{F}^Q_{\cE_1}} \Tr(P_a \rho_{\cE_2}) \\
    &= \sum_{\substack{f_{\cE_1} \in M^p_{4\epsilon}(\cF_f)}} \sum_{a: h_{\mathrm{Q}, a} \in \mathcal{F}^Q_{\cE_1}} \Tr(P_a \rho_{\cE_2}) - \sum_{a: h_{\mathrm{Q}, a} \in \mathcal{F}^Q_{\cE_2}} \Tr(P_a \rho_{\cE_2}) \\
    &= \sum_{\substack{a: \exists f_{\cE_1} \in M^p_{4\epsilon}(\cF_f) \\ h_{\mathrm{Q}, a} \in \mathcal{F}^Q_{\cE_1}}} \Tr(P_a \rho_{\cE_2}) - \sum_{a: h_{\mathrm{Q}, a} \in \mathcal{F}^Q_{\cE_2}} \Tr(P_a \rho_{\cE_2}) \\
    &\leq \sum_{a=1}^A \Tr(P_a \rho_{\cE_2}) - \sum_{a: h_{\mathrm{Q}, a} \in \mathcal{F}^Q_{\cE_2}} \Tr(P_a \rho_{\cE_2}) \\
    &= 1 - P_{\cE_2, \cE_2} \leq 1/3.
\end{align}
The first equality uses the definition of matrix $P$. The third equality follows from the observation that distinct clusters are also disjoint ($\mathcal{F}^Q_{\cE_1} \cap \mathcal{F}^Q_{\cE_2} = \varnothing$).
We conclude that the $|M^p_{4\epsilon}(\cF_f)| \times |M^p_{4\epsilon}(\cF_f)|$-matrix $P$
is diagonally dominant. Such matrices are guaranteed to be non-singular, i.e.\ they have full rank.

This is a suitable starting point for analyzing the probability $\Tr(P_a \rho_{\cE})$ via a polynomial method \cite{beals2001quantum}.
Let
 $\{K^{\cE}_i\}_{i=1}^{2^n 2^m}$,
\begin{equation}
    {\cE}(\rho) = \sum_{i} K^{\cE}_i \rho (K^{\cE}_i)^\dagger,
\end{equation}
with $K^{\cE}_i \in \mathbb{C}^{2^m \times 2^n}$ be the Kraus representation of a fixed CPTP map $\cE \in \cF$.
This representation is parametrized by (at most) $2^n 2^m \times 2^m  2^n=2^{2(n+m)}$ complex parameters:
\begin{equation*}
(K^{\cE}_{i})_{jk} = z^{\cE}_{i \times 2^{n+m} + j \times 2^{n} + k}
\quad \text{defines} \quad z^{\cE} \in \mathbb{C}^{2^{2(n+m)}}.
\end{equation*}
On a high level, we parametrize inputs to the quantum ML model by vectors.
After training, the probability of obtaining answer $a \in A$ corresponds to a homogeneous polynomial of degree $N_{\mathrm{Q}}$ in $z^{\cE}$ and of degree $N_{\mathrm{Q}}$ in $\bar{z}^{\cE}$:
\begin{align}
    \Tr(P_a \rho_{\cE}) &= \Tr(P_a \mathcal{C}_{N_{\mathrm{Q}}} (\cE \otimes\, \mathcal{I}) \mathcal{C}_{N_{\mathrm{Q}}-1} \ldots \mathcal{C}_2 (\cE \otimes\, \mathcal{I}) \mathcal{C}_1 (\cE \otimes\, \mathcal{I}) \mathcal{C}_0(\rho_0))\\
    &= w_a^\dagger  \underset{N_{\mathrm{Q}}\text{ times}}{\underbrace{(z^{\cE} \otimes \bar{z}^{\cE}) \otimes \cdots \otimes (z^{\cE} \otimes \bar{z}^{\cE}})},
\end{align}
where $w_a^\dagger$ is a dual tensor product vector with compatible dimension $N=2^{2(n+m)2N_{\mathrm{Q}}}=2^{4(n+m)N_{\mathrm{Q}}}$. Every matrix element $P_{\mathcal{E}_1,\mathcal{E}_2}$ of $P$ defined in Equation~\eqref{eq:polynomial-method-P} can be expressed as a sum of homogeneous polynoials in $(z^{\mathcal{E}_2} \otimes \bar{z}^{\mathcal{E}_2}) \otimes \cdots \otimes (z^{\mathcal{E}_2} \otimes \bar{z}^{\mathcal{E}_2})$.
Collecting all $M=|M^p_{4\epsilon}(\cF_f)|$ possible tensor products as rows of the matrix
\begin{equation}
Z = \begin{bmatrix}
(z^{\mathcal{E}_1} \otimes \bar{z}^{\mathcal{E}_1}) \otimes \cdots \otimes (z^{\mathcal{E}_1} \otimes \bar{z}^{\mathcal{E}_1})
& \cdots &
(z^{\mathcal{E}_M} \otimes \bar{z}^{\mathcal{E}_M}) \otimes \cdots \otimes (z^{\mathcal{E}_M} \otimes \bar{z}^{\mathcal{E}_M})
\end{bmatrix}
\in \mathbb{C}^{N \times M}
\end{equation}
allows us to present the multilinear characterization of all entries of $P$ in a single display:
\begin{align*}
P= W Z \quad \text{with} \quad  W
= \begin{bmatrix}
\sum_{h_{\mathrm{Q}, a} \in \mathcal{F}^Q_{\cE_1}} w_a^\dagger \\
\vdots \\
    \sum_{h_{\mathrm{Q}, a} \in \mathcal{F}^Q_{\cE_M}} w_a^\dagger \\
    \end{bmatrix}
    \in \mathbb{C}^{M \times 2^N}
\end{align*}
Above, we have shown that the $M \times M$-matrix $P$ must have full column rank. This is only possible if
\begin{equation}
|M^p_{4\epsilon}(\cF_f)| = M \leq N = 2^{2(n+m)2N_{\mathrm{Q}}}=2^{4(n+m)N_{\mathrm{Q}}}.
\end{equation}
Rearranging these terms and assuming $n \leq m$ implies the following lower bound on quantum query complexity $N_{\mathrm{Q}}$:
\begin{equation}
    N_{\mathrm{Q}} \geq \frac{\log(|M^p_{4\epsilon}(\cF_f)|)}{4(n + m)} \geq \frac{\log(|M^p_{4\epsilon}(\cF_f)|)}{8m}.
\end{equation}

\subsection{Information-theoretic upper bound for restricted classical machine learning models} \label{sec:uppbdCL}

We will focus on restricted classical ML models that can select inputs $x_i \in \{0, 1\}^n$ and obtain the corresponding outcome $o_i \in \mathbb{R}$. This outcome is obtained by performing a single-shot measurement (the projective measurement given by the eigenbasis of $O$)
of observable $O$ on the output quantum state
 $\mathcal{E}(\lvert x_i \rangle \! \langle x_i \rvert)$.
This ensures
\begin{equation}
\E[o_i] = \Tr(O \mathcal{E}(\lvert x_i \rangle \! \langle x_i \rvert) =f_{\cE}(x_i) \quad \text{and, moreover,} \quad \left|o_i \right| \leq 1 \text{ with probability one}, \label{eq:target-fct-thm1}
\end{equation}
because observables are bounded in spectral norm ($\norm{O} \leq 1$).
By using the obtained training data $\{(x_i, o_i)\}_{i}$, the restricted classical ML model will produce a prediction model $h_{\mathrm{C}}(x)$ that allows accurate prediction of $f_{\cE}(x) = \Tr(O \cE(\ketbra{x}{x}))$.
The restricted classical ML model should provide accurate prediction model for any CPTP map $\cE \in \cF$.

\subsubsection{Classical machine learning model for a given learning problem}

We consider the following classical machine learning model. First, we sample $N$ classical inputs $x_1, \ldots, x_{N}$ according to the distribution $\mathcal{D}$.
Then, we obtain an associated quantum measurement outcome $o_i$ for each input $x_i$.That is, $o_i$ is a random variable that reproduces the target function in expectation only, see Equation~\eqref{eq:target-fct-thm1}. We denote the underlying distribution by $\mathcal{D}_o (O,\cE(\ketbra{x_i}{x_i}))$ to delineate dependence on input $x_i$ and CPTP map $\cE$.
After obtaining the training data $\{(x_i, o_i)\}_{i=1}^N$, the model performs the following optimization to minimize the \emph{empirical} training error
\begin{equation}
    f_* = \argmin_{f \in M^p_{4 \epsilon}(\cF_f)} \frac{1}{N} \sum_{i=1}^N |f(x_i) - o_i|^2. \label{eq:training-step}
\end{equation}
Here, $M^p_{4 \epsilon}(\cF_f)$ is the maximal packing net defined in Section~\ref{sec:packnetQ}.
The packing net is a subset of the set $\cF_f$ that contains functions that are sufficiently different from one another.
By ``empirical training error'' we mean the deviation of the function $f(x_i)$ from the actual measurement outcome $o_i$, averaged over $N$ data points:
$\frac{1}{N} \sum_{i=1}^N |f(x_i) - o_i|^2.$
In the later discussion, we will also refer to the \emph{ideal} training error, meaning the average deviation of the function $f(x_i)$ from the expectation value $f_{\cE}(x_i) = \Tr(O \cE(\ketbra{x_i}{x_i}))$:
$\frac{1}{N} \sum_{i=1}^N |f(x_i) - f_{\cE}(x_i)|^2.$
This distinction is important; the ideal training error could be close to zero as long as the maximal packing net $M^p_{4 \epsilon}(\cF_f)$ is closely packed, but because of the statistical fluctuation in the quantum measurements, the outcomes $\{o_i\}$ can deviate substantially from the expectation value $f_{\cE}(x_i)$. Therefore we might not be able to achieve small empirical training error even if $f = f_{\cE}$.
%

In the following, we will provide a tight statistical analysis for bounding the prediction error
\begin{equation}
    \E_{x \sim \mathcal{D}} |f_*(x) - f_{\cE}(x)|^2.
\end{equation}
The statistical analysis relies crucially on the distance measure used to define the packing net $M^p_{4 \epsilon}(\cF_f)$.
Recall that this is the average squared distance over the input distribution $\mathcal{D}$, and the statistical fluctuation in performing quantum measurements to obtain $o_i$, see Equation~\eqref{eq:epspackingnetdistance}.
In particular, we will show that a data size of $N = \Theta(\log(|M^p_{4 \epsilon}(\cF_f)|) / \epsilon)$ suffices to achieve prediction errors of order $\mathcal{O}(\epsilon)$ only.

We find it worthwhile to point out that this scaling is better than one might expect.
Standard results in statistical learning theory \cite{bartlett2002rademacher, mohri2018foundations} usually yield a data size of order $\log(|M^p_{4 \epsilon}(\cF_f)|) / \epsilon^2$, which is worse than our result by an additional $1/\epsilon$ factor.

\subsubsection{Concentration results I: Ideal training error}

We begin by considering the concentration of the \emph{ideal} training error for an arbitrary function $f$ from the maximal packing net $M^p_{4 \epsilon}(\cF_f)$:
\begin{equation}
\frac{1}{N} \sum_{i=1}^N |f(x_i) - f_{\cE}(x_i)|^2,
\end{equation}
which only depends on the inputs $x_1, \ldots, x_N$ and is independent of the observable measurement outcome $o_i$.
We use the quantifier \emph{ideal} because we compare directly with the expectation value $f_{\cE}(x_i)$ rather than the measurement outcome $o_i$.

As a first step, view $\left|f(x) - f_{\cE}(x) \right|^2$ with $x \overset{\mathcal{D}}{\sim} \left\{0,1\right\}^n$ as a random variable and check that it is bounded: $\left|f(x)-f_{\cE} (x)\right|^2 \leq (|f(x)| + |f_{\cE} (x)|)^2 \leq 4$ for all $x \in \left\{0,1\right\}^n$. This implies the following bound on the variance:
\begin{align}
\Var[|f(x) - f_{\cE}(x)|^2]
\leq \E_{x \sim \mathcal{D}}|f(x) - f_{\cE}(x)|^4
\leq 4 \E_{x \sim \mathcal{D}}|f(x) - f_{\cE}(x)|^2.
\end{align}
We see that the ideal training error as a sum of independent random variables with bounded variance.
Bernstein's inequality implies for $t>0$
\begin{align}
\mathrm{Pr} \left[
\left| \frac{1}{N} \sum_{i=1}^N |f(x_i) - f_{\cE}(x_i)|^2
 - \E_{x \sim \mathcal{D}} |f(x) - f_{\cE}(x)|^2 \right| \geq t \right]
\leq 2 \exp\left( - \frac{1}{2} \frac{N t^2}{ 4 \E_{x \sim \mathcal{D}} |f(x) - f_{\cE}(x)|^2 + \frac{8}{3} t} \right). 
\end{align}
Assigning $t = \frac{1}{4} \E_{x \sim \mathcal{D}} |f(x) - f_{\cE}(x)|^2$ allows us to conclude
\begin{align}
&\mathrm{Pr} \left[
\left| \frac{1}{N} \sum_{i=1}^N |f(x_i) - f_{\cE}(x_i)|^2
 - \E_{x \sim \mathcal{D}} |f(x) - f_{\cE}(x)|^2 \right| \geq \frac{1}{4} \E_{x \sim \mathcal{D}} |f(x) - f_{\cE}(x)|^2 \right]\\
\leq & 2 \exp\left( - \frac{3}{448} N \E_{x \sim \mathcal{D}} |f(x) - f_{\cE}(x)|^2 \right)
\end{align}
On the other hand, if
 $\E_{x \sim \mathcal{D}} |f(x) - f_{\cE}(x)|^2 \leq 4 \epsilon$, we instead assign $t = \epsilon$ to obtain
\begin{align}
&\mathrm{Pr} \left[
\left| \frac{1}{N} \sum_{i=1}^N |f(x_i) - f_{\cE}(x_i)|^2
 - \E_{x \sim \mathcal{D}} |f(x) - f_{\cE}(x)|^2 \right| \geq \epsilon \right]
\leq 2 \exp\left( - \frac{3}{112} N \epsilon \right).
\end{align}
These two tail bounds (that cover different regimes) and a union bound then imply
\begin{align}
&\mathrm{Pr} \left[ \forall f \in M^p_{4 \epsilon}(\cF_f),
\left| \frac{1}{N} \sum_{i=1}^N |f(x_i) - f_{\cE}(x_i)|^2
 - \E_{x \sim \mathcal{D}} |f(x) - f_{\cE}(x)|^2 \right| \geq \frac{1}{4} \max\left(4 \epsilon, \E_{x \sim \mathcal{D}} |f(x) - f_{\cE}(x)|^2 \right) \right] \nonumber\\
\leq & 2 \sum_{f \in M^p_{4 \epsilon}(\cF_f)} \exp\left( - \frac{3}{448} N \max\left(4 \epsilon, \E_{x \sim \mathcal{D}} |f(x) - f_{\cE}(x)|^2 \right) \right)
\leq  2 \left|M^p_{4 \epsilon}(\cF_f)\right| \exp\left( - \frac{3}{112} N \epsilon \right).
\end{align}
Intuitively, this can be understood as follows.
The functions $f$ close to $f_{\cE}$ will be distorted by at most $\epsilon$, while the functions $f$ that are further away from $f_{\cE}$ will be distorted by a value proportional to the distance $\E_{x \sim \mathcal{D}} |f(x) - f_{\cE}(x)|^2$.
For $\delta \in (0,1)$ (confidence), we set
\begin{equation} \label{eq:cmlNbound}
    N \geq \frac{38 \log(2 \left|M^p_{4 \epsilon}(\cF_f)\right|/ \delta)}{\epsilon}.
\end{equation}
(Throughout this paper, $\log$ has base $\mathrm{e}$ unless otherwise indicated.)
Then, with probability at least $1 - \delta$, we have
\begin{equation} \label{eq:fluc-f-fcE}
    \forall f \in M^p_{4 \epsilon}(\cF_f),
\left| \frac{1}{N} \sum_{i=1}^N |f(x_i) - f_{\cE}(x_i)|^2
 - \E_{x \sim \mathcal{D}} |f(x) - f_{\cE}(x)|^2 \right| < \frac{1}{4} \max\left(4 \epsilon, \E_{x \sim \mathcal{D}} |f(x) - f_{\cE}(x)|^2 \right).
\end{equation}
We note that the training data size $N$ in Equation~(\ref{eq:cmlNbound}) scales as $1 / \epsilon$ rather than $1 / \epsilon^2$, an improvement over the standard scaling typically encountered in
statistical learning theory \cite{bartlett2002rademacher, mohri2018foundations}.
The $1 / \epsilon^2$ comes naturally when we sample over the different inputs $x_i$ and apply a concentration inequality on the ideal training error $\frac{1}{N} \sum_{i=1}^N |f(x_i) - f_{\cE}(x_i)|^2$ to guarantee an  $\mathcal{O}(\epsilon)$ statistical fluctuation around the prediction error $\E_{x \sim \cD} |f(x) - f_{\cE}(x)|^2$.
The main reason for the improved scaling is that any function $f$ with a small prediction error $\E_{x \sim \mathcal{D}} |f(x) - f_{\cE}(x)|^2 \leq 4 \epsilon$ also has a small variance (e.g., a highly biased coin that almost always come out heads has a variance close to zero, which is much smaller than for an unbiased coin), so we need only $N=\mathcal{O}(1/\epsilon)$ to achieve $\mathcal{O}(\epsilon)$ statistical fluctuations. Furthermore, by examining Equation~\eqref{eq:fluc-f-fcE}, we see that if function $f$  has a large prediction error then the statistical fluctuations in the training data may also be large.
This is a price we pay to avoid a training set with size $N$ scaling as $1/\epsilon^2$.
The increased statistical fluctuations for a function $f$ with a large prediction error are not problematic, because the statistical fluctuation are still smaller than the prediction error in that case.
We find that functions with small prediction error have small ideal training error, while functions with large prediction error have large ideal training error, which is adequate for our purposes.

We condition on the event that the display Equation~(\ref{eq:fluc-f-fcE}) holds true and proceed to the second step.

\subsubsection{Concentration results II: Shifted empirical training error}

In the second step, we will condition on a set of inputs $x_1, \ldots, x_N$ and study the concentration of statistical fluctuations in the observable measurement outcome $o_i$.
Let us define a new quantity, which we call the \emph{shifted} empirical training error:
\begin{equation}
\frac{1}{N} \sum_{i=1}^N \left[ |f(x_i) - o_i|^2 - |f_{\cE}(x_i) - o_i|^2 \right],
\end{equation}
where $f$ can be any function in the packing net $M^p_{4 \epsilon}(\cF_f)$.
The expectation value of the shifted empirical training error can be computed by means of direct expansion.
Use $f_{\cE}(x_i)= \E_{o \sim \mathcal{D}_o(O, \mathcal{E}(\lvert x_i \rangle \! \langle x_i \rvert))} [o] $
to rewrite
\begin{equation}
    |f(x_i) - f_{\cE}(x_i)|^2=\E_{\substack{o \sim \mathcal{D}_o(O, \mathcal{E}(\lvert x_i \rangle \! \langle x_i \rvert))}} |f(x_i) - o|^2 - |f_{\cE}(x_i) - o|^2,
\end{equation}
and for fixed input $x_i$, we can also bound the variance:
\begin{align}
    \Var_{o \sim \mathcal{D}_o(O, \mathcal{E}(\lvert x_i \rangle \! \langle x_i \rvert))} |f(x_i) - o|^2 - |f_{\cE}(x_i) - o|^2 &= \E_{\substack{o \sim \mathcal{D}_o(O, \mathcal{E}(\lvert x_i \rangle \! \langle x_i \rvert))}} 4 (f(x_i) - f_{\cE}(x_i))^2 (o - f_{\cE}(x_i))^2\\
    &=  4 (f(x_i) - f_{\cE}(x_i))^2 \Var_{o \sim \mathcal{D}_o(O, \mathcal{E}(\lvert x_i \rangle \! \langle x_i \rvert))}[o] \\
    &\leq  4 |f(x_i) - f_{\cE}(x_i)|^2.
\end{align}
The last inequality is contingent on $\|O \| \leq 1$ which implies $o \in \left[-1,1\right]$ with probability one.
Now, we apply Bernstein's inequality again. The $o_i$'s are independent, bounded random variables with small variance. So, we obtain
\begin{align}
&\mathrm{Pr} \left[
\left| \frac{1}{N} \sum_{i=1}^N \left[ |f(x_i) - o_i|^2 - |f_{\cE}(x_i) - o_i|^2 \right]
 - \frac{1}{N} \sum_{i=1}^N |f(x_i) - f_{\cE}(x_i)|^2 \right| \geq t \right]\\
\leq &2 \exp\left( - \frac{1}{2} \frac{N t^2}{\frac{4}{N} \sum_{i=1}^N |f(x_i) - f_{\cE}(x_i)|^2 + \frac{8}{3} t} \right) \quad \text{for $t >0$}.
\end{align}
Assigning $t = \frac{1}{4N} \sum_{i=1}^N |f(x_i) - f_{\cE}(x_i)|^2$ ensures
\begin{align}
&\mathrm{Pr} \left[
\left| \frac{1}{N} \sum_{i=1}^N \left[ |f(x_i) - o_i|^2 - |f_{\cE}(x_i) - o_i|^2 \right]
 - \frac{1}{N} \sum_{i=1}^N |f(x_i) - f_{\cE}(x_i)|^2 ] \right| \geq \frac{1}{4N} \sum_{i=1}^N |f(x_i) - f_{\cE}(x_i)|^2 \right]\\
\leq & 2 \exp\left( - \frac{3}{448} \sum_{i=1}^N |f(x_i) - f_{\cE}(x_i)|^2 \right). \label{eq:foconc-large}
\end{align}
If $\frac{1}{N} \sum_{i=1}^N |f(x_i) - f_{\cE}(x_i)|^2 \leq 4 \epsilon$, we instead assign $t = \epsilon$ to obtain
\begin{align}
\mathrm{Pr} \left[
\left| \frac{1}{N} \sum_{i=1}^N \left[ |f(x_i) - o_i|^2 - |f_{\cE}(x_i) - o_i|^2 \right]
 - \frac{1}{N} \sum_{i=1}^N |f(x_i) - f_{\cE}(x_i)|^2 ] \right| \geq \epsilon \right]
\leq 2 \exp\left( - \frac{3}{112} N \epsilon \right). \label{eq:foconc-small}
\end{align}
These results are conditioned on $x_1, \ldots, x_N$ already being sampled (the result of first step) where the event given in Equation~\eqref{eq:fluc-f-fcE} holds.

\subsubsection{Prediction error for functions in the maximal packing net}

Before bounding the prediction error of $f_\star$ obtained by the restricted classical ML model, we need to show that the following events happen simultaneously with high probability.
\begin{itemize}
    \item \emph{Event 1:} There exists a function $\tilde{f} \in M^p_{4 \epsilon}(\cF_f)$ with small prediction error that results in an empirical training error that is upper bounded by a certain threshold.
    In particular, we will set out to show that
    \begin{equation}
    \frac{1}{N} \sum_{i=1}^N |\tilde{f}(x_i) - o_i|^2 \leq \frac{1}{N} \sum_{i=1}^N |f_{\cE}(x_i) - o_i|^2 + \frac{25}{4} \epsilon.
    \label{eq:classical-ML-aux2}
    \end{equation}
    \item \emph{Event 2:} All functions $f \in M^p_{4 \epsilon}(\cF_f)$ that have a large prediction error will result in an empirical training error lower bounded by a certain threshold.
    In particular, we define the event to be: for all models $f \in M^p_{4 \epsilon}(\cF_f)$ such that $\E_{x \sim D} |f(x) - f_{\cE}(x)|^2 \geq 12 \epsilon$ will have:
    \begin{equation}
    \frac{1}{N} \sum_{i=1}^N |f(x_i) - o_i|^2  \geq \frac{1}{N} \sum_{i=1}^N |f_{\cE}(x_i) - o_i|^2 + \frac{27}{4} \epsilon.
    \label{eq:classical-ML-aux1}
    \end{equation}
\end{itemize}
Then, we can combine these statements to obtain a bound on the prediction error of $f_\star$. Let us first relate the packing net to another useful concept.

\begin{lemma}[Maximal packing nets are covering nets] \label{lem:maxpackingnet}
For all $f_{\cE} \in \cF_f$, there exists $\tilde{f} \in M^p_{4 \epsilon}(\cF_f)$, such that
\begin{equation}
    \E_{x \sim D} |f(x) - f_{\cE}(x)|^2 \leq 4 \epsilon.
\end{equation}
\end{lemma}

The proof is standard, see e.g.\ \cite{vershynin2018high}, and based on contradicting the assumption that $M^p_{4\epsilon}(\cF_f)$ is a maximal packing net. Since it is short and insightful, we include the full proof for completeness. 

\begin{proof}[Proof of Lemma~\ref{lem:maxpackingnet}]
If there exists $f_{\cE} \in \cF$, such that for all $\tilde{f} \in M^p_{4 \epsilon}(\cF_f)$, we have
\begin{equation}
    \E_{x \sim \mathcal{D}} |f(x) - f_{\cE}(x)|^2 > 4 \epsilon,
\end{equation}
then we can add $f_{\cE}$ into the packing net $M^p_{4 \epsilon}(\cF_f)$. Hence $M^p_{4 \epsilon}(\cF_f)$ is not the maximal packing net.
\end{proof}

For Event 1 in Equation~\eqref{eq:classical-ML-aux2}, we want to show the existence of a function $\tilde{f}$ that has a small prediction error as well as an empirical training error upper bounded by a threshold.
Because $M^p_{4 \epsilon}(\cF_f)$ is a maximal packing net, using Lemma~\ref{lem:maxpackingnet}, there exists a function $\tilde{f}$ such that the prediction error
\begin{equation}
    \E_{x \sim D} |\tilde{f}(x) - f_{\cE}(x)|^2 \leq 4 \epsilon.
\end{equation}
We now condition on Equation~\eqref{eq:fluc-f-fcE} being true, which happens with probability at least $1-\delta$.
Therefore, we have the following bound on the ideal training error
\begin{equation} \label{eq:emppredicerr-5/4}
    \frac{1}{N} \sum_{i=1}^N |\tilde{f}(x_i) - f_{\cE}(x_i)|^2 \leq 5 \epsilon.
\end{equation}
We can now use this insight to control the shifted empirical training error. Use a combination of Eqs.~\eqref{eq:foconc-large}~and~\eqref{eq:foconc-small} to conclude
\begin{align}
&\mathrm{Pr} \left[
\left| \frac{1}{N} \sum_{i=1}^N \left[ |\tilde{f}(x_i) - o_i|^2 - |f_{\cE}(x_i) - o_i|^2 \right]
 - \frac{1}{N} \sum_{i=1}^N |\tilde{f}(x_i) - f_{\cE}(x_i)|^2 ] \right| \geq \frac{5}{4} \epsilon \right] \\
\leq & 2 \exp\left( - \frac{3}{112} N \epsilon \right) \leq \frac{\delta}{|M^p_{4 \epsilon}(\cF_f)|}.
\end{align}
The first inequality comes from separately analyzing the two cases: $\frac{1}{N} \sum_{i=1}^N |\tilde{f}(x_i) - f_{\cE}(x_i)|^2 \leq 4 \epsilon$ or $4 \epsilon < \frac{1}{N} \sum_{i=1}^N |\tilde{f}(x_i) - f_{\cE}(x_i)|^2 \leq 5 \epsilon$, then take the looser statement.
The second inequality arises from inserting the (lower bound) on training data size $N$ from Equation~\eqref{eq:cmlNbound}.
Therefore, if the display from Equation~\eqref{eq:fluc-f-fcE} is true (which happens with probability at least $1-\delta$), then
\begin{align}
    \frac{1}{N} \sum_{i=1}^N |\tilde{f}(x_i) - o_i|^2 &\leq \frac{1}{N} \sum_{i=1}^N |f_{\cE}(x_i) - o_i|^2 + \frac{1}{N} \sum_{i=1}^N |\tilde{f}(x_i) - f_{\cE}(x_i)|^2 + \frac{5}{4} \epsilon\\
    &\leq \frac{1}{N} \sum_{i=1}^N |f_{\cE}(x_i) - o_i|^2 + \frac{25}{4} \epsilon \quad \text{with probability at least $1 - \delta / |M^p_{4 \epsilon}(\cF_f)|$}.
\end{align}
The second inequality is contingent on using Equation~\eqref{eq:emppredicerr-5/4}.
This is Event $1$ that we have set out to establish. And it is guaranteed to happen with high probability.

We now move on to Event $2$ given in Equation~\eqref{eq:classical-ML-aux1}. For any $f \in M^p_{4 \epsilon}(\cF_f)$ with a large prediction error
\begin{equation}
    \E_{x \sim D} |f(x) - f_{\cE}(x)|^2 \geq 12 \epsilon,
\end{equation}
we want to show that the training error $\frac{1}{N} \sum_{i=1}^N |f(x_i) - o_i|^2$ will also be large.
We again condition on the event displayed by Equation~\eqref{eq:fluc-f-fcE} (which happens with probability at least $1-\delta$).
This relation implies the following bound on the ideal training error
\begin{equation}
    \frac{1}{N} \sum_{i=1}^N |f(x_i) - f_{\cE}(x_i)|^2 \geq \frac{3}{4} \E_{x \sim D} |f(x) - f_{\cE}(x)|^2 \geq 9 \epsilon.
\end{equation}
Using the concentration result from Equation~\eqref{eq:foconc-large}, we have
\begin{align}
&\mathrm{Pr} \left[
\left| \frac{1}{N} \sum_{i=1}^N \left[ |f(x_i) - o_i|^2 - |f_{\cE}(x_i) - o_i|^2 \right]
 - \frac{1}{N} \sum_{i=1}^N |f(x_i) - f_{\cE}(x_i)|^2 ] \right| \geq \frac{1}{4N} \sum_{i=1}^N |f(x_i) - f_{\cE}(x_i)|^2 \right]\\
&\leq 2 \exp\left( - \frac{3}{448} \sum_{i=1}^N |f(x_i) - f_{\cE}(x_i)|^2 \right) \leq 2 \exp\left( - \frac{3}{448} N \left(9 \epsilon\right) \right) \leq \frac{\delta}{|M^p_{4 \epsilon}(\cF_f)|}.
\end{align}
The last inequality uses the training data size bound from Equation~\eqref{eq:cmlNbound}.
This ensures
\begin{align}
    \frac{1}{N} \sum_{i=1}^N |f(x_i) - o_i|^2 &\geq \frac{1}{N} \sum_{i=1}^N |f_{\cE}(x_i) - o_i|^2 + \frac{3}{4} \frac{1}{N} \sum_{i=1}^N |f(x_i) - f_{\cE}(x_i)|^2\\
    &\geq \frac{1}{N} \sum_{i=1}^N |f_{\cE}(x_i) - o_i|^2 + \frac{9}{16} \E_{x \sim D} |f(x) - f_{\cE}(x)|^2\\
    &\geq \frac{1}{N} \sum_{i=1}^N |f_{\cE}(x_i) - o_i|^2 + \frac{27}{4} \epsilon \quad \text{with probability at least $1 - \delta / |M^p_{4 \epsilon}(\cF_f)|$}.
\end{align}
We can combine these insights by applying union bound to obtain that all the desired events given in Equation~\eqref{eq:classical-ML-aux2}~and~\eqref{eq:classical-ML-aux1} happen simultaneously with probability at least $1 - \delta$ if we condition on Equation~\eqref{eq:fluc-f-fcE} to be true.
Furthermore, because the event in display~\eqref{eq:fluc-f-fcE} happens with probability $1 - \delta$, we can guarantee that the the desired events given in Equation~\eqref{eq:classical-ML-aux2}~and~\eqref{eq:classical-ML-aux1} happen simultaneously with a probability at least $(1 - \delta)^2 \geq 1 - 2 \delta$. This statement uses the following basic fact from elementary probability theory:
Let $p({A|B})$ be the probability of event $A$ when we condition on event $B$. And let $p(B)$ be the probability of event $B$. Then the probability of event $A$, $p(A)$, is larger or equal to the probability that $A, B$ both happens, $p(A, B) = p(A|B) p(B)$.

To conclude, we have shown that using a training data of size $N\geq 38 \log(2 \left|M^p_{4 \epsilon}(\cF_f)\right|/ \delta) / \epsilon$ guarantees that the relations given in Equation~\eqref{eq:classical-ML-aux2}~and~\eqref{eq:classical-ML-aux1} happen with probability at least $1 - 2 \delta$.

\subsubsection{Prediction error for functions produced by restricted classical ML}

Let us choose a data of size $N \geq 38 \log(4 \left|M^p_{4 \epsilon}(\cF_f)\right|/ \delta) / \epsilon$ such that the two relations in Equation~\eqref{eq:classical-ML-aux2}~and~\eqref{eq:classical-ML-aux1} are both true with probability $1-\delta$. We now combine these with two other concepts from the previous subsections.
Let $f_{\cE}(x)$ be the actual target function and recall that (at least) one packing net element $\tilde{f} \in M^p_{4 \epsilon}(\cF_f)$ is guaranteed to be close ($\E_{x \sim \mathcal{D}} | \tilde{f}(x) - f_{\cE} (x) |^2 \leq 4 \epsilon$ according to Lemma~\ref{lem:maxpackingnet}).
The restricted classical ML model tries to identify such a packing net element by minimizing the empirical training error: $f_* = \argmin_{f \in M^p_{4 \epsilon}(\cF_f)} \frac{1}{N} \sum_{i=1}^N |f(x_i) - o_i|^2$ according to Equation~\eqref{eq:training-step}.
This setup ensures
\begin{align}
\frac{1}{N} \sum_{i=1}^N |f_*(x_i) - o_i|^2
= \min_{f \in M^{p}_{4\epsilon} (\cF_f)} \frac{1}{N} \sum_{i=1}^N |f(x_i) - o_i|^2
\leq \frac{1}{N} \sum_{i=1}^n \left|\tilde{f} (x_i) - o_i \right|^2.
\end{align}
The first relation in Equation~\eqref{eq:classical-ML-aux2} allows us to take it from there. Indeed,
\begin{align*}
\frac{1}{N} \sum_{i=1}^n \left|\tilde{f} (x_i) - o_i \right|^2
\leq \frac{1}{N} \sum_{i=1}^N |f_{\cE}(x_i) - o_i|^2 + \frac{25}{4} \epsilon
< \frac{1}{N} \sum_{i=1}^N |f_{\cE}(x_i) - o_i|^2 + \frac{27}{4} \epsilon,
\end{align*}
where the strict inequality is completely trivial. Apply the second relation in Equation~\eqref{eq:classical-ML-aux1} to complete the chain of arguments:
\begin{align}
 \frac{1}{N} \sum_{i=1}^N |f_*(x_i) - o_i|^2 <
 \min_{f \in M_{4\epsilon}^p (\cF_f):\; \E_{x \sim \mathcal{D}}|f(x)-f_{\cE}(x)|^2 \geq 12 \epsilon}
 \frac{1}{N} \sum_{i=1}^N \left| f(x_i) - o_i \right|^2.
\end{align}
In words, this display implies that the empirical training error achieved by $f_*$ -- the output of the restricted classical ML model -- is strictly smaller than any empirical training error that could be achieved by any packing net function that has a comparatively large prediction error (at least $12 \epsilon$).
Therefore if $f_*$ has a prediction error of at least $12 \epsilon$, then this leads to a contradiction that $\frac{1}{N} \sum_{i=1}^N |f_*(x_i) - o_i|^2 < \frac{1}{N} \sum_{i=1}^N |f_*(x_i) - o_i|^2$.
By contradiction, this claim implies that the prediction error achieved by $f_*$ cannot be too bad. 

\begin{proposition} \label{prop:uppbdcml}
Let $f_*: \left\{0,1\right\}^n \to \mathbb{R}$ be the packing net element that minimizes the empirical training error in Equation~\eqref{eq:training-step}. Then, for $\delta \in (0,1)$, training data of size $N \geq 38 \log(4 \left|M^p_{4 \epsilon}(\cF_f)\right|/ \delta) / \epsilon$ implies:
\begin{equation}
\E_{x \sim \mathcal{D}} |f^*(x) - f_{\cE}(x)|^2
< 12 \epsilon \quad \text{with probability at least $1-\delta$.}
\end{equation}
\end{proposition}

\subsection{Combining the upper and lower bound} \label{sec:combineCLQU}

If a quantum ML model produces a prediction $h_{\mathrm{Q}}$ achieving average prediction error
\begin{equation}
    \E_{x \sim \mathcal{D}} \left| h_{\mathrm{Q}}(x) - \Tr(O \mathcal{E}(\lvert x \rangle \! \langle x \rvert)) \right|^2 \leq \epsilon,
\end{equation}
with probability at least $2/3$ for any CPTP map $\mathcal{E} \in \cF$, then,
as proven in Equation~\eqref{eq:NQlowerbound}, the quantum ML must access the map $\cE$ at least $N_{\mathrm{Q}}$ times, where
\begin{equation}
     N_{\mathrm{Q}} =\Omega\left(\frac{\log(|M^p_{4\epsilon}(\cF_f)|)}{m}\right).
\end{equation}
On the other hand, from Proposition~\ref{prop:uppbdcml}, we know there is a restricted classical ML model producing prediction $h_{\mathrm{C}}$ achieving average prediction error
\begin{equation}
    \E_{x \sim \mathcal{D}} \left| h_{\mathrm{C}}(x) - \Tr(O \mathcal{E}(\lvert x \rangle \! \langle x \rvert)) \right|^2 \leq 12 \epsilon = \mathcal{O}(\epsilon),
\end{equation}
with high probability for any CPTP map $\cE \in \cF$, such that the restricted classical ML accesses the map $N_{\mathrm{C}}$ times, where
\begin{equation}
 N_{\mathrm{C}} = \mathcal{O}\left( \frac{\log(|M^p_{4\epsilon}(\cF_f)|)}{\epsilon}\right) = \mathcal{O}\left( \frac{m N_{\mathrm{Q}}}{\epsilon} \right).
\end{equation}
This concludes the proof of Theorem~\ref{thm:noadvquantum}.

\section{Examples saturating the maximum information-theoretic advantage} \label{sec:satnoadv}

\begin{proposition} \label{prop:noadvquantum-appendix}
For any $\epsilon \in (0,1/3)$, and positive integer $m$, there exists a learning problem \eqref{eq:target-function} 
-- specified by an $m$-qubit observable $O$, a set $\mathcal{F}$ of CPTP maps and a distribution $\mathcal{D}$ on $n$-bit inputs, where $n=m{-}1$, -- with the following property:
Any restricted classical ML model that can learn a function $h_{\mathrm{C}}(x)$ that achieves
\begin{equation}
\E_{x \sim \mathcal{D}} \left| h_{\mathrm{C}}(x) - \Tr(O \mathcal{E}(\ketbra{x}{x})) \right|^2 \leq \epsilon,
\end{equation}
must use classical training data of size $N_{\mathrm{C}} = \Omega(m N_{\mathrm{Q}} / \epsilon)$, where $N_{\mathrm{Q}}$ is the number of queries in the best quantum ML model.
\end{proposition}

This statement follows from constructing a stylized learning problem that admits the largest possible separation (albeit only a small polynomial factor). We first introduce the problem and discuss quantum and classical strategies (and their limitations) afterwards. We will focus on restricted classical ML models, because Theorem~\ref{thm:noadvquantum} also considers restricted classical ML models.
We leave open the question of whether the separation between unrestricted classical ML and quantum ML is tight or not.

\paragraph{Learning problem formulation}
Fix $\epsilon \in (0,1/3)$, let $m$ be the integer in the statement of Proposition \ref{prop:noadvquantum-appendix} and set $n=m-1$.
We consider a set of CPTP maps $\mathcal{F}=\left\{ \mathcal{E}_{a}:\; a \in \left\{0,1\right\}^{n}\right\}$ containing $2^{n}$ elements, where each map in the set takes an $n$-qubit input to an $(n{+}1)$-qubit output.
The map $\mathcal{E}_a$, labeled by bit string $a \in \left\{0,1\right\}^{n}$, is comprised of $2 \times 2^{n}$ Kraus operators:
\begin{equation}
\mathcal{E}_{a}(\rho)
= \sum_{z \in \left\{0,1\right\}^{n}}
\sum_{i=1}^2 K^{z,i}_{a}\rho (K^{z,i}_{a})^\dagger
\quad \text{with} \quad
\begin{cases}
K_{a}^{z,1} =& \sqrt{\tfrac{1+\sqrt{3 \epsilon}}{2}} (I \otimes I^{\otimes n})|a \odot z,a \rangle \! \langle z|, \\
K_{a}^{z,2} =& \sqrt{\tfrac{1-\sqrt{3 \epsilon}}{2}} (X \otimes I^{\otimes n})|a \odot z,a \rangle \! \langle z|.
\end{cases}
\label{eq:function-class}
\end{equation}
Here, $X$ is a single-qubit bit flip and
$a \odot z \in \left\{0,1\right\}$ denotes the inner product of bit-strings in $\mathbb{Z}_2$.
We also choose the $(n{+}1)$-qubit observable
$O= Z \otimes I^{\otimes n}$, i.e.\ we measure the first qubit in the $Z$-basis and trace out the rest of the system. By construction, the resulting function 
admits a closed-form expression:
\begin{equation}
f_a (x) = \Tr \left( O \cE_{a} (\ketbra{x}{x}) \right)
=
\sqrt{3 \epsilon}\left( 1- 2 a \odot x \right).
\label{eq:closed-form-expression}
\end{equation}
We consider $\cD$ to be the uniform distribution over the $n$-bit inputs.

\paragraph{Upper bound on the quantum query complexity}

The above learning problem is 
easy to solve in the quantum realm. Since the set of CPTP maps $\mathcal{F} = \left\{ \mathcal{E}_a:\; a \in \left\{0,1\right\}^{n}\right\}$
is known, it suffices to extract the label $a \in \left\{0,1\right\}^{n}$ of the underlying CPTP map. Once $a$ is known, the closed-form expression~\eqref{eq:closed-form-expression} allows to predict future function values $f_a (x)$ efficiently with perfect accuracy -- regardless of the input $x \in \left\{0,1\right\}^n$.

Quantum computers are well equipped to extract the label $a$. In fact, a single query of the unknown CPTP map $\mathcal{E}_a$ suffices to extract the label by executing the following simple procedure:
\begin{enumerate}
\item prepare the all-zero state on $n$ qubits: $\rho_0 = |0,\ldots,0 \rangle \! \langle 0,\ldots,0|$;
\item query $\mathcal{E}$ and apply it to $\rho_0$: $\rho_1 = \tfrac{1}{2} (I + \sqrt{3 \epsilon} Z ) \otimes |a \rangle \! \langle a|$, according to Equation~\eqref{eq:function-class};
\item throw away (trace out) the first qubit to obtain the $n$ remaining ones: $\rho_2 = |a \rangle \! \langle a|$;
\item perform a computational basis measurement to extact $a \in \left\{0,1\right\}^{n}$ with probability one.
\end{enumerate}
We see that a single quantum 
query ($N_{\mathrm{Q}}=1$) suffices to extract the label $a$ with certainty. Subsequently, we can make efficient and perfect predictions via the closed-form expression~\eqref{eq:closed-form-expression}: 
\begin{equation}
\mathbb{E}_{x \sim \mathcal{D}} \left| h_{\mathrm{Q}}(x)-\Tr \left( O \cE_a (\ketbra{x}{x}) \right) \right|^2 =0 \leq \epsilon \quad \Leftarrow \quad N_{\mathrm{Q}} =1. \label{eq:explicit-quantum-bound}
\end{equation}
In words, $N_{\mathrm{Q}}=1$ allows for training a quantum ML model $h_{\mathrm{Q}}(x)$ that achieves zero prediction error
for all input distributions (perfect prediction).
This concrete ML model is also optimal, because $N_{\mathrm{Q}}=1$ is the smallest number of queries conceivable ($N_{\mathrm{Q}}=0$ would not reveal any information about the underlying CPTP map).

\paragraph{Lower bound on the classical query complexity}

Let us now turn to potential classical strategies for solving the above learning problem. In contrast to the previous paragraph, we will not construct an explicit strategy. Instead, we will use ideas similar to Appendix~\ref{sec:MIanalysisML} to establish a fundamental lower bound.

Recall that the input distribution $\cD$ is taken to be the uniform distribution.
Also, for each $\cE_a \in \cF$, the underlying function $f_a (x) = \Tr \left(O\mathcal{E}_a (\ketbra{x}{x})\right)$ admits a closed form expression, see Equation~\eqref{eq:closed-form-expression}.
For $a,b \in \left\{0,1\right\}^{n}$,
\begin{align}
\E_{x \sim \mathcal{D}} \left|f_a (x) - f_b (x) \right|^2
=&
\frac{1}{2^{n}}\sum_{x \in \left\{0,1\right\}^{n}}  \left| \sqrt{3 \epsilon} 2 (a - b) \odot x \right|^2 
= \begin{cases}
0 & \text{if $a=b$}, \\
6 \epsilon& \text{else},
\end{cases}
\end{align}
because $2^{-n}\sum_{x \in \left\{0,1\right\}^{n}} \left|c \odot x \right|^2 = 1/2$ for all $n$-bit strings $c \neq (0,\ldots,0)$.
Now, suppose that a restricted classical ML model can utilize training data $\mathcal{T}=\left\{(x_i,o_i)\right\}_{i=1}^{N_{\mathrm{C}}}$ to learn a function $h_{\mathrm{C}}(x)$ that obeys $\E_{x \sim \mathcal{D}} \left|h_{\mathrm{C}} (x) - \Tr \left(O \cE_a (\ketbra{x}{x})\right) \right|^2 \leq \epsilon$ with high probability for any label $a \in \left\{0,1\right\}^{n}$.
Then, this model would also allow us to identify the underlying label.
Indeed $\E_{x \sim \mathcal{D}} \left|h_{\mathrm{C}} (x) - f_b  \right|^2 \leq \epsilon$ if $b=a$, while for $b \neq a$, by the triangle inequality,
\begin{equation}\label{eq:triangle-distinguish}
\E_{x \sim \mathcal{D}} \left|h_{\mathrm{C}} (x) - f_b(x)  \right|^2
\geq \Big(\sqrt{\E_{x \sim \mathcal{D}} \left|f_a (x) - f_b (x) \right|^2}
- \sqrt{\E_{x \sim \mathcal{D}} \left|h_{\mathrm{C}} (x) - f_a (x) \right|^2} \Big)^2
\geq \left(\sqrt{6 \epsilon} - \sqrt{\epsilon}\right)^2 
> \epsilon.
\end{equation}
By checking $\E_{x \sim D} \left|h_{\mathrm{C}}(x)-f_b (x)\right|^2 \leq \epsilon$ for every possible value $b \in \left\{0,1\right\}^n$, the restricted classical ML model allows us to recover the underlying bit-string label $a \in \left\{0,1\right\}^{n}$
with high probability. For this part of the argument, what's essential is that the right-hand-side of the inequality Equation~(\ref{eq:triangle-distinguish}) is greater than $\epsilon$. If we replace $3\epsilon$ in Equation~(\ref{eq:function-class}) by $\alpha \epsilon$, where $\alpha$ is a constant, we require $\sqrt{2\alpha} - 1> 1$, or $\alpha > 2$. We chose $\alpha = 3$ merely for convenience.

For any random hidden bitstring $a \in \{0, 1\}^n$, we can use the restricted classical ML to obtain the training data $\{(x_i, o_i)\}_{i=1}^{N_{\mathrm{C}}}$ and determine the underlying bitstring $a$.
We assume that the restricted classical ML first query $x_1$ obtains $o_1$, then query $x_2$ obtains $o_2$, and so on.
We also have
\begin{equation}
o_i =
\begin{cases}
+ 1& \text{with probability } p_+=\frac{1}{2} \left( 1 + \sqrt{3 \epsilon} \left(1- 2 a \odot x_i \right) \right), \\
-1 &  \text{with probability } p_-=\frac{1}{2} \left( 1- \sqrt{3 \epsilon} \left(1- 2 a \odot x_i \right) \right),
\end{cases}
\label{eq:measurement-outcome}
\end{equation}
which is a single-shot outcome for measuring the observable $O$ on the state $\cE_a (\ketbra{x_i}{x_i})$ in the eigenbasis of $O=Z \otimes I^{\otimes n}$.
Because we can use the training data $\{(x_i, o_i)\}_{i=1}^{N_{\mathrm{C}}}$ to determine $a$ with high probability (by the assumption of the restricted classical ML model), Fano's inequality and the data processing inequality then imply a bound on the mutual information between the training data and the CPTP map label $a$:
\begin{equation}
I \big( a: \left\{ (x_i,o_i) \right\}_{i=1}^{N_{\mathrm{C}}} \big) = \Omega(n).
\label{eq:classical-aux1}
\end{equation}
Next, using chain rule of mutual information based on conditional mutual information, we have
\begin{equation}
I \big( a: \left\{ (x_i,o_i) \right\}_{i=1}^{N_{\mathrm{C}}} \big)
= \sum_{i=1}^{N_{\mathrm{C}}} I (a: (x_i,o_i) | \{(x_j, o_j)\}_{j=1}^{i-1}) = \sum_{i=1}^{N_{\mathrm{C}}} I (a: o_i | \{(x_j, o_j)\}_{j=1}^{i-1}, x_i).
\label{eq:classical-aux2}
\end{equation}
The second equality follows from the fact that $x_i$ is chosen by the restricted classical ML using only the information of $\{(x_j, o_j)\}_{j=1}^{i-1}$, hence the input $x_i$ does not provide any additional information about $a$, i.e., $I(a : x_i | \{(x_j, o_j)\}_{j=1}^{i-1}) = 0$. We now upper bound each term:
\begin{equation}
    I (a: o_i | \{(x_j, o_j)\}_{j=1}^{i-1}, x_i) = H(o_i | \{(x_j, o_j)\}_{j=1}^{i-1}, x_i) - H(o_i | \{(x_j, o_j)\}_{j=1}^{i-1}, x_i, a)
\end{equation}
Because $o_i$ is a two-outcome random variable, $H(o_i | \{(x_j, o_j)\}_{j=1}^{i-1}, x_i) \leq H(o_i) \leq \log_2(2)=1$.
We now consider the distribution of $o_i$ when we condition on
$\{(x_j, o_j)\}_{j=1}^{i-1}, x_i, a$.
A closer inspection of Equation~\eqref{eq:measurement-outcome} reveals that the probability of one outcome is $p=\tfrac{1}{2} \left( 1 + \sqrt{3 \epsilon} \right)$ and the other is $1-p$ (the value $a \odot x_i \in \left\{0,1\right\}$ only ever permutes the outcome sign). This ensures
\begin{equation}
H(o_i | \{(x_j, o_j)\}_{j=1}^{i-1}, x_i, a) = -p \log_2 (p) - (1-p) \log_2 (1-p) \geq \log_2 (2) - (2p-1)^2,
\end{equation}
and we conclude
\begin{equation}
I (a: o_i | \{(x_j, o_j)\}_{j=1}^{i-1}, x_i) \leq (2p-1)^2 = 3 \epsilon.
\label{eq:classical-aux3}
\end{equation}
Finally, we 
combine Eqs.~\eqref{eq:classical-aux1}, \eqref{eq:classical-aux2} and \eqref{eq:classical-aux3} to conclude
\begin{align*}
\Omega (n) \leq  I \big( a: \left\{ (x_i,o_i) \right\}_{i=1}^{N_{\mathrm{C}}} \big)
\leq \sum_{i=1}^{N_{\mathrm{C}}} I (a: o_i | \{(x_j, o_j)\}_{j=1}^{i-1}, x_i) \leq 3 \epsilon N_{\mathrm{C}}.
\end{align*}
Therefore, recalling that the output size of our set of maps is $m=n{+}1$, we have for a restricted classical ML model with small average prediction error:
\begin{equation}
\E_{x \sim \mathcal{D}} \left| h_{\mathrm{C}}(x) - \Tr \left( O \cE_a (\ketbra{x}{x}) \right) \right|^2 \leq \epsilon \quad \text{with high probability} \quad \Rightarrow \quad N_{\mathrm{C}} = \Omega \left( m/\epsilon \right).
\end{equation}
Proposition~\ref{prop:noadvquantum-appendix} follows from combining this assertion with the fact that the underlying learning problem does admit a perfect quantum solution with $N_{\mathrm{Q}}=1$, see Equation~\eqref{eq:explicit-quantum-bound}.

\section{Exponential separation for predicting expectation values of Pauli operators}
\label{app:exp-sep-Pauli}

In this section, we consider an example to demonstrate the existence of exponential information-theoretic quantum advantage when we want to achieve small worst-case prediction error.

\subsection{Task description}
\label{sec:exp-task-desc}

The learning task is to train an ML model that allows accurate prediction of
\begin{equation}
    x \in \{I, X, Y, Z\}^{n} \rightarrow \Tr(P_x \rho),
\end{equation}
where $\rho$ is an unknown $n$-qubit state, $X, Y, Z$ are the single-qubit Pauli operators, and $P_x$ is the tensor product of Pauli operators given by $x$.
This task does fit 
into the framework described in the main text.
Every unknown quantum state $\rho$ defines an unknown quantum channel $\cE_\rho$.
The unknown quantum channel $\cE_\rho$ takes a classical input $x \in \{I, X, Y, Z\}^n$, prepares the unknown quantum state $\rho$, and rotates the quantum state $\rho$ according to the input $x$ such that a Pauli-Z measurement on the first qubit is equivalent to measuring $P_x$ on $\rho$.
More precisely, we define
\begin{equation}
    \cE_\rho(\ketbra{x}{y}) = \delta_{x, y} C_x^\dagger \rho C_x,
\end{equation}
where $\ket{x}$ is an encoding of the classical input $x$ (e.g., as a $2n$-qubit computational basis state), $C_x$ is a Clifford unitary that satisfies
\begin{equation}
    C_x Z_1 C_x^\dagger = P_x.
\end{equation}
We can extend this definition linearly to all of quantum state space.
The goal of the machine learning model is to produce $f(x)$ such that
\begin{equation}
    \max_{x \in \{I, X, Y, Z\}^n} \left| f(x) - \Tr(Z_1 \cE_\rho(\ketbra{x}{x})) \right| \leq \epsilon.
\end{equation}
We consider restricted classical ML models that can only obtain $\{(x_i, o_i)\}$, where $x_i$ are inputs denoting the Pauli operators, and $o_i$ are measurement outcome when measuring $\cE_\rho(\ketbra{x_i}{x_i})$ with $Z_1$, which is equivalent to the measurement outcome of measuring $\rho$ with $P_{x_i}$.
Classical ML models can obtain $\{(x_i, o_i)\}$, where $o_i$ is now the measurement outcome of an arbitrary POVM on the output state $\cE_\rho(\ketbra{x_i}{x_i})$.
Note that for restricted classical ML, $o_i \in \mathbb{R}$, but for classical ML, $o_i$ is in the set indexing the POVM elements $\{F_o\}_o$ with $\sum_o F_o = I$.
Quantum ML models can access the unknown quantum channel $\cE_{\rho}$ at will.
Recall that, according to Theorem \ref{thm:noadvquantum}, for achieving a small average prediction error according to any input distribution, no large quantum advantage in sample complexity can be found.

In the following, we will show that to achieve a small worst-case prediction error, an exponential quantum advantage in sample complexity is possible.
In Section~\ref{sec:pauli-qml}, we give a simple quantum ML algorithm that can accurately predict the expectation values of all $4^n$ Pauli observables using only $\mathcal{O}(n)$ samples.
In Section~\ref{sec:lowbd-cml-pauli}~and~\ref{sec:lowerboundindmeas}, we will show that any classical ML algorithm requires at least an exponential number of samples to accurately predict the expectation values of all $4^n$ Pauli observables.
In Section~\ref{sec:lowerboundentangled}, we will give a matching lower bound $\Omega(n)$ for quantum ML.

\subsection{Sample complexity of a quantum ML algorithm}
\label{sec:pauli-qml}

The quantum ML algorithm accesses the quantum channel $\cE_\rho$ multiple times to obtain multiple copies of the underlying quantum state $\rho$.
Each access to $\cE_\rho$ allows us to obtain one copy of $\rho$.
Then, the quantum ML algorithm performs a sequence of measurements on the copies of $\rho$ to accurately predict $\Tr(P_x \rho)$ for all $x \in \{I, X, Y, Z\}^n$.
To this end, we will give a detailed proof for Theorem~\ref{thm:sampPauliQ} given below.
From Theorem~\ref{thm:sampPauliQ}, we only need $\mathcal{O}(\log(100 \times 4^n) / \epsilon^4) = \mathcal{O}(n / \epsilon^4)$ copies to accurately predict all $4^n$ Pauli operators with probability at least $0.99$.
Hence we only need to access the quantum channel $\cE_\rho$ for $\mathcal{O}(n / \epsilon^4)$ times.

\begin{theorem} \label{thm:sampPauliQ}
For any $M$ Pauli operators $P_1, \ldots, P_M$, there is a procedure that produces $\hat{p}_1, \ldots, \hat{p}_M$ with
\begin{equation}
    \left|\hat{p}_i - \Tr(P_i \rho)\right| \leq \epsilon,\,\, \forall i = 1, \ldots, M,
\end{equation}
under probability at least $1-\delta$ by performing POVM measurements on
\begin{equation}
    N = \mathcal{O}\left( \frac{\log(M / \delta)}{\epsilon^4} \right)
\end{equation}
copies of the unknown quantum state $\rho$.
\end{theorem}

The procedure takes in the Pauli operators one by one and produces the estimate $\hat{p}_i$ sequentially.
Throughout the prediction process, the procedure maintains two blocks of memory:
\begin{enumerate}
    \item Classical memory: We perform $N_1$ repetitions of Bell measurements on two copies of the quantum state $\rho \otimes \rho$. For repetition $t$, we go through every qubit, and measure the $k$-th qubit from the first and second copies in the Bell basis to obtain:
\begin{equation}\label{eq:Bell-basis-meas}
    S^{(t)}_k \in \Big\{ \ketbra{\Psi^+}{\Psi^+}, \ketbra{\Psi^-}{\Psi^-}, \ketbra{\Phi^+}{\Phi^+}, \ketbra{\Phi^-}{\Phi^-} \Big\},
\end{equation}
where the Bell basis encompasses four maximally entangled 2-qubit states. Set $\ket{\Omega} = \tfrac{1}{\sqrt{2}} \left( \ket{00}+\ket{11} \right)$ (Bell state) and define
\begin{align*}
    \ket{\Psi^+} &= I \otimes I \ket{\Omega}
    =\frac{1}{\sqrt{2}} \left( \ket{00} + \ket{11} \right),\\
    \ket{\Psi^-} &=   I \otimes Z \ket{\Omega}
    = \frac{1}{\sqrt{2}} \left( \ket{00} - \ket{11} \right),\\
   \ket{\Phi^+} &= I \otimes X \ket{\Omega}
    = \frac{1}{\sqrt{2}} \left( \ket{01} + \ket{10} \right),\\
    \ket{\Phi^-} &= \mathrm{i} I \otimes Y \ket{\Omega}
    =\frac{1}{\sqrt{2}} \left( \ket{01} - \ket{10} \right).
\end{align*}
We then efficiently store the measurement data $S^{(t)}_k, \forall k = 1, \ldots, n, \forall t = 1, \ldots, N_1$ in a classical memory with $2 nN_1$ bits.
We use this block of memory to estimate $|\Tr(P \rho)|^2$ for any Pauli operator $P$.
    \item Quantum memory: We store $N_2$ copies of the unknown quantum state $\rho$. We use this block of memory to estimate $\sign(\Tr(P \rho))$ for any Pauli operator $P$.
\end{enumerate}

Consider any tensor product of Pauli operators $P = \sigma_{1} \otimes \ldots \otimes \sigma_{n}$, where $\sigma_k \in \{I, X, Y, Z\}, \forall k = 1, \ldots, n$. The classical memory is used to predict the absolute value of $\Tr(P \rho)$, while the quantum memory is used to predict the sign of $\Tr(P \rho)$. This allows us to obtain an accurate estimate for $\Tr(P \rho)$. The following remark is central to the procedure.

\begin{remark}
If we find that the absolute value of $\Tr(P \rho)$ is close to zero, then we need not 
use the quantum memory to predict the sign of $\Tr(P \rho)$. By measuring the sign only for the Pauli observables such that the absolute value of the expectation value is appreciable, we can find the sign for may Pauli observables without badly disturbing the copies of $\rho$ stored in the quantum memory.
\end{remark}

\noindent We proceed to give a detailed procedure for estimating the absolute value and the sign of $\Tr(P_i \rho)$.

\subsubsection{Measuring absolute values}

To understand how the absolute value of the expectation of a Pauli operator is estimated, first consider the case where $\rho$ is the density operator of a single qubit, and suppose that $\rho\otimes \rho$ is measured in the Bell basis. The outcome $S$ is the projector onto one of the four Bell states, as in Equation~(\ref{eq:Bell-basis-meas}). If $\sigma\in \{I,X,Y,Z\}$ is any Pauli matrix, then each Bell state is an eigenstate of $\sigma\otimes\sigma$ with eigenvalue $\pm 1$. In the state $S$, the $+1$ eigenvalue of $\sigma\otimes\sigma$ occurs with probability $\mathrm{Prob}(+) = \frac{1}{2}\Tr\left((I\otimes I+\sigma\otimes\sigma)(\rho\otimes\rho)\right)$, and the $-1$ eigenvalue  occurs with probability $\mathrm{Prob}(-) = \frac{1}{2}\Tr\left((I\otimes I-\sigma\otimes\sigma)(\rho\otimes\rho)\right)$. Therefore,
\begin{equation}
    \E\left[\Tr\left((\sigma \otimes \sigma) S\right)\right]
    = \mathrm{Prob}(+) - \mathrm{Prob}(-) = \Tr \left((\sigma\otimes\sigma)(\rho\otimes\rho)\right) = |\Tr (\sigma\rho)|^2,
\end{equation}

This observation can be generalized to the case where $\rho$ is an $n$-qubit state, and each pair of qubits in $\rho\otimes \rho$ is measured in the Bell basis, yielding the outcomes $\{S_k, k = 1,2, \dots, n\}$. If $P = \sigma_{1} \otimes \ldots \otimes \sigma_{n}$ is a Pauli observable, then $S_k$ is an eigenstate of $\sigma_k\otimes \sigma_k$ with eigenvalue $\pm 1$ for each $k$, just as in the $n=1$ case; in particular,
\begin{equation}
     \prod_{k=1}^n \Tr\left((\sigma_{k} \otimes \sigma_{k}) S_k\right)
     =\pm 1.
\end{equation}
This product is $+1$ when  $\otimes_{k=1}^n S_k$ is an eigenstate of $P\otimes P$ with eigenvalue $+1$, and it is $-1$ when  $\otimes_{k=1}^n S_k$ is an eigenstate of $P\otimes P$ with eigenvalue $-1$.
Therefore,
\begin{align}\label{eq:bell-expected}
     & \E\left[\prod_{k=1}^n \Tr\left((\sigma_{k} \otimes \sigma_{k}) S_k\right)\right]
     =\E\left[ \Tr\left( (P \otimes P) \bigotimes_{k=1}^n S_k \right) \right]\nonumber\\
    & = \mathrm{Prob}(P\otimes P =+1) - \mathrm{Prob}(P\otimes P=-1)
      =\Tr\left( (P \otimes P) (\rho \otimes \rho) \right)
     = |\Tr(P \rho)|^2.
\end{align}

Because Equation~(\ref{eq:bell-expected}) relates the distribution of Bell measurement outcomes to $|\Tr(P \rho)|$, we can estimate $|\Tr(P \rho)|$ accurately by repeating Bell measurement on $\rho\otimes\rho$ sufficiently many times. Suppose we have altogether $2N_1$ copies of $\rho$, and perform the Bell measurement on $N_1$ pairs of copies. We collect the measurement data $\{S_k^{(t)}\}$ in the classical memory, where $k=1,2,\dots n$ labels the qubit pairs, and $t=1,2, \dots N_1$ labels the repeated measurements.

For each Pauli observable $P = \sigma_{1} \otimes \ldots \otimes \sigma_{n}$, we define the corresponding expression
\begin{equation}
    \hat{a}(P) = \frac{1}{N_1} \sum_{t = 1}^{N_1} \prod_{k=1}^n \Tr\left((\sigma_{k} \otimes \sigma_{k}) S^{(t)}_k\right),
\end{equation}
which can be computed efficiently in time  $\mathcal{O}(nN_1)$.

Using the statistical property given in
Equation~(\ref{eq:bell-expected}), we can apply Hoeffding's inequality to
show that, with high probability, the estimate $\hat{a}(P)$ is close to the expectation value $|\Tr(P \rho)|^2$:
\begin{lemma}\label{lem:hoeffPauli}
Given $N_1 = \Theta(\log(1 / \delta) / \epsilon^2)$. For any Pauli operator $P$, we have
\begin{equation}
    \left|\hat{a}(P) - |\Tr(P \rho)|^2\right| \leq \epsilon,
\end{equation}
with probability at least $1 - \delta$.
\end{lemma}
\noindent To obtain an estimate for the absolute value $|\Tr(P \rho)|$, we consider the following estimate
\begin{equation}
    \hat{b} = \sqrt{\max(0, \hat{a})}.
\end{equation}
It is not hard to show the following implication using $\sqrt{x + y} \leq \sqrt{x} + \sqrt{y}$,
\begin{align}
    |\Tr(P \rho)|^2 - \epsilon \leq \hat{a} \leq |\Tr(P \rho)|^2 + \epsilon
    \implies \max(0, \sqrt{|\Tr(P \rho)|^2} - \sqrt{\epsilon}) \leq \hat{b} \leq \sqrt{|\Tr(P \rho)|^2} + \sqrt{\epsilon}.
\end{align}
Therefore, Lemma~\ref{lem:hoeffPauli} gives the following corollary. With high probability, we can estimate the absolute value of $\Tr(P \rho)$ for any Pauli operator $P$ accurately.
\begin{corollary} \label{cor:absPauli}
Given $N_1 = \Theta(\log(1 / \delta) / \epsilon^4)$. For any Pauli operator $P$, we have
\begin{equation}
    \left|\hat{b} - |\Tr(P \rho)|\right| \leq \epsilon,
\end{equation}
with probability at least $1 - \delta$.
\end{corollary}

\subsubsection{Measuring signs}

Given a Pauli operator $P$, we check if $|\Tr(P \rho)|$ is large enough with the previous procedure using Bell measurement and a the classical memory.
If $|\Tr(P \rho)|$ is large enough, we proceed to measure the sign of $\Tr(P \rho)$ as well, using a quantum memory. Suppose that $N_2$ copies of the state $\rho$ are stored in the quantum memory. The sign is determined by measuring the two outcome observable
\begin{equation}
    E = \sum_{z \in \{\pm 1\}^{N_2}} \mathrm{MAJ}(z) \, \Pi^{(z_1)} \otimes \ldots \otimes \Pi^{(z_{N_2})},
\end{equation}
where $\Pi^{(z_k)}$ projects the $k$th copy onto the eigenspace of the Pauli operator $P$ with eigenvalue $z_k\in\{+1,-1\}$ and $\mathrm{MAJ}(z)$ is the majority vote. In effect, the observable $E$ measures $P$ on each of the $N_2$ copies of $\rho$, obtaining either $+1$ or $-1$ each time, and takes a majority vote on these outcomes, yielding +1 if more than half of the outcomes are +1, and yielding -1 if more than half of the outcomes are -1.

Intuitively, if we are guaranteed that $|\Tr(P \rho)| > \tilde{\epsilon}$ and $N_2$ is sufficiently large, then the majority vote will concentrate on the correct answer. Hence, the quantum state $\rho^{\otimes N_2}$ is approximately contained in one of the two eigenspaces of the observable $E$. As a result, after measuring $E$, the quantum state $\rho^{\otimes N_2}$ would remain approximately the same.
This allows us to keep measuring the sign of $\Tr(P \rho)$ for many different Pauli operators.
This strategy is a key element in the original protocol for shadow tomography \cite{aaronson2018shadow}.
The rigorous guarantee is given by Lemma~\ref{eq:signmeasure}, which relies on the quantum union bound \cite{aaronson2006qma}.

\begin{lemma}[Quantum union bound \cite{aaronson2006qma}] \label{lem:Qunion}
Given any quantum state $\rho$. Consider a sequence of $M$ two-outcome observables $\{K_i\}_{i=1}^M$, where $K_i$ has eigenvalue $0$ or $1$. Assume $\Tr(K_i \rho) \geq 1 - \epsilon$. When we measure $K_1, \ldots, K_M$ sequentially on $\rho$, the probability that all of them yield the outcome $1$ is at least $1 - M \sqrt{\epsilon}$.
\end{lemma}

\begin{lemma} \label{eq:signmeasure}
For any $\delta, \epsilon, M > 0$, let $N_2 = \Theta(\log(M / \delta) / \epsilon^2)$.
For any $M$ Pauli operators $P_1, \ldots, P_M$ with
\begin{equation}
    \left|\Tr(P_i \rho)\right| > \epsilon, \forall i =1, \ldots, M,
\end{equation}
let us define the corresponding two-outcome observables $E_1, \ldots, E_M$.
If we measure the two-outcome observable $E_1, \ldots, E_M$ sequentially on $\rho^{\otimes N_2}$, then we can correctly obtain
\begin{equation}
    \sign(\Tr(P_i \rho)), \forall i =1, \ldots, M,
\end{equation}
with probability at least $1 - \delta$.
\end{lemma}
\begin{proof}
Let us first consider the probability that $E_i$ outputs the sign of $\Tr(P_i \rho)$ when measuring on $\rho^{\otimes N_2}$.
Using Hoeffding's inequality, the probability is lower bounded by
\begin{equation}
    1 - \mathrm{e}^{-\frac{1}{2} |\Tr(P_i \rho)|^2 N_2} > 1 - \mathrm{e}^{-\frac{1}{2} \epsilon^2 N_2} = 1 - \frac{\delta^2}{2 M^2},
\end{equation}
if we take $N_2 = 4 \log(2 M / \delta) / \epsilon^2$.
Let us define a related observable $K_i$ which has outcome $1$ if the outcome of $E_i$ is equal to the sign of $\Tr(P_i \rho)$ and $0$ otherwise.
We have the following bound on the expectation value:
\begin{align}
    \Tr(K_i \rho^{\otimes N_2}) \geq 1 - \frac{\delta^2}{M^2}.
\end{align}
Note that measuring $K_i$ is the same as measuring $E_i$. The only difference is in the eigenvalue associated with the outcome ($E_i$ has outcome $\pm 1$, while $K_i$ has outcome $0, 1$).
Using the quantum union bound in Lemma~\ref{lem:Qunion}, we will obtain outcome $1$ when we measure $K_i$ for all $i = 1, \ldots, M$ with probability at least $1 - M \sqrt{\delta^2 / M^2} = 1 - \delta$.
Because measuring outcome $1$ when we measure $K_i$ is equivalent to obtaining the correct sign when we measure $E_i$, this concludes the proof.
\end{proof}

\subsubsection{Sample complexity analysis}

We can combine the previous results to obtain the sample complexity $N_1, N_2$ to guarantee accurate prediction of $\Tr(P_i \rho), \forall i = 1, \ldots, M$. Here by ``sample complexity'' we mean the number of copies of $\rho$ consumed by the protocol.
First, following Corollary~\ref{cor:absPauli} and the union bound, we choose
\begin{equation}
N_1 = \Theta\left(\frac{\log(M / \delta)}{\epsilon^4}\right)
\end{equation}
such that with probability at least $1 - (\delta / 2)$, we have
\begin{equation} \label{eq:goodabs}
    \left| \hat{b}_i - |\Tr(P_i \rho)| \right| \leq \epsilon / 3, \forall i = 1, \ldots, M.
\end{equation}
This allows accurate prediction for the absolute values.
In the next step, we only obtain the sign of $\Tr(P_i \rho)$ if $\hat{b}_i > (2/3) \epsilon$.
Let $R$ denote the number of Pauli operators that achieve $\hat{b}_i > (2/3) \epsilon$.
Conditioning on Event~\eqref{eq:goodabs}, for all $P_i$ with $\hat{b}_i > (2/3)\epsilon$, we have $\Tr(P_i \rho) > (1/3) \epsilon$.
Let us denote the measured signs to be $\hat{s}_i, \forall i = 1, \ldots, M$. If we do not measure the sign for $P_i$, then we set $\hat{s}_i = 0$.
Using Lemma~\ref{eq:signmeasure}, we can choose a number of samples
\begin{equation}
N_2 = \Theta\left(\frac{\log(R / \delta)}{\epsilon^2}\right) \leq \mathcal{O}\left(\frac{\log(M / \delta)}{\epsilon^2}\right)
\end{equation}
to guarantee that, with probability at least $1 - (\delta / 2)$, the measured signs are all correct for all $P_i$ with $\hat{b}_i > (2/3)\epsilon$:
\begin{equation} \label{eq:correctsign}
    \hat{s}_i = \sign(\Tr(P_i \rho)), \forall i \in \{1, \ldots, M\}: \hat{b}_i > \frac{2}{3}\epsilon.
\end{equation}
Together, with probability at least $(1 - (\delta / 2))^2 \geq 1 - \delta$, Events~\eqref{eq:goodabs}~and~\eqref{eq:correctsign} both holds.
Finally, we produce the following estimate for $\Tr(P_i \rho)$:
\begin{equation}
    \hat{p}_i = \begin{cases}
    \hat{b}_i \hat{s}_i, & \text{ if } \hat{b}_i > \frac{2}{3} \epsilon,\\
    0, & \text{ else}.
    \end{cases}
\end{equation}
For Pauli operator $P_i$ with $\hat{b}_i > \frac{2}{3} \epsilon$, we have
\begin{equation}
    \left|\hat{p}_i - \Tr(P_i \rho)\right| = \left|\hat{s}_i\left( \hat{b}_i - |\Tr(P_i \rho)| \right)\right| = \left|\hat{b}_i - |\Tr(P_i \rho)|\right| \leq \epsilon /3 \leq \epsilon.
\end{equation}
For Pauli operator $P_i$ with $\hat{b}_i \leq \frac{2}{3} \epsilon$, we have
\begin{equation}
    \left|\hat{p}_i - \Tr(P_i \rho)\right| = \left|\Tr(P_i \rho)\right| \leq \left|\hat{b}_i\right| + \left||\Tr(P_i \rho)| - \hat{b}_i\right| \leq \frac{2}{3} \epsilon + \frac{1}{3} \epsilon \leq \epsilon.
\end{equation}
The first inequality uses triangle inequality. The second inequality uses the assumption that $\hat{b}_i \leq \frac{2}{3} \epsilon$ and the fact that $\hat{b}_i \geq 0$.
Hence, we successfully predict the expectation value of $\Tr(P_i \rho), \forall i = 1, \ldots, M$.
The number of copies we used is
\begin{equation}
    N = 2 N_1 + N_2 = \mathcal{O}\left( \frac{\log(M  / \delta)}{\epsilon^4} \right).
\end{equation}
This concludes the proof of Theorem~\ref{thm:sampPauliQ}.

\subsubsection{Heuristics for near-term implementations}
\label{sec:heuristicsNTImp}

The procedure for measuring the absolute value of $\Tr(P \rho)$ requires only two-copy Bell basis measurements, which can be performed quite easily on current quantum devices \cite{cotler2019quantum, huggins2020virtual}.
On the other hand, the procedure for measuring signs of $\Tr(P \rho)$ can be rather difficult to implement on a near-term quantum device.
However, for measuring signs, we only need to investigate the Pauli operator expectation values found to be relatively large in absolute value. For any $P$ such that $|\Tr(P \rho)|$ is small,
rather than determine the sign we simply set our predicted expectation value of $P$ to zero.

Therefore we focus on the  Pauli observables whose expectation values are (comparatively) large in absolute value.
When two Pauli observables $P_1$ and $P_2$ anti-commute, the expectation values of $P_1$ and $P_2$ cannot both be large in absolute value.
Hence, it is often the case that these remaining observables can be sorted into a collection of just a few sets, where operators in each set are mutually commuting.
Various sorting strategies are known \cite{crawford2020efficient, izmaylov2019unitary, verteletskyi2020measurement, hamamura2020efficient, jiang2020optimal, huggins2019efficient, bonet2020nearly}.
Once such a commuting set is identified, all the operators in the set can be simultaneously measured using just a single copy of $\rho$,

As a simple illustrative example, suppose that the underlying state is a stabilizer state. Even if we consider all $4^n$ Pauli observables, when we filter out all Pauli observables with $\Tr(P \rho) = 0$, the rest of the Pauli observables will form a single commuting set.
Hence, we only need to measure in one appropriately chosen basis to estimate all non-zero Pauli observables simultaneously.

\subsection{Sample complexity lower bound for restricted classical ML algorithms}
\label{sec:lowbd-cml-pauli}

Recall that the restricted classical ML algorithm can only choose a certain input $x$ and obtain measurement outcome $o$ when we measure a fixed observable $O = Z_1$ on $\cE_{\rho}(\ketbra{x}{x})$.
The sample complexity lower bound can be proved by reducing the problem to a well known classical problem: learning point functions. This section establishes a sample complexity lower bound for this basic problem.

\begin{lemma} \label{lem:learnpointf}
Consider a set of point functions $f_a: \{0, 1\}^n \rightarrow \{0, 1\}$, where $n \geq 2, a \in \{0, 1\}^n$ and
\begin{equation}
    f_a(x) = \begin{cases}
    1 & \text{if } x = a,\\
    0 & \text{if } x \neq a.\\
    \end{cases}
\end{equation}
Suppose a classical randomized algorithm can output a function $\tilde{f}$ by obtaining data for $N$ different inputs $\{(x_i, f_a(x_i))\}_{i=1}^N$, such that
\begin{equation}
    \max_{x \in \{0, 1\}^n} \left|\tilde{f}(x) - f_a(x)\right| < \frac{1}{2}\quad \text{with probability at least $2/3$},
    \label{eq:point-function}
\end{equation}
 for all $a \in \{0, 1\}^n$. Then $N \geq (1/4) 2^n= \Omega(2^n)$.
\end{lemma}
\begin{proof}
Suppose $a$ is selected uniformly at random.
By assumption, the classical randomized algorithm will be able to output $\tilde{f}$ that obeys Equation~\eqref{eq:point-function}.
Because the worst-case prediction error is smaller than $1/2$, the classical randomized algorithm will be able to correctly identify $a$ with probability at least $2/3$.
We now use a simple calculation to obtain a lower bound to the probability for any classical randomized algorithm to identify $a$.
Because $f_a(x)=0$ for all but one inputs, a uniformly random input $x_1$ will result in $f_a(x_1) = 0$ with probability $(2^n - 1) / (2^n)$.
Conditioned on $f_a(x_1) = 0$, the second query $x_2$ will result in $f_a(x_2) = 0$ with probability $(2^n - 2) / (2^n - 1)$.
By induction, the probability that the first $k$ queries all result in function value equal to $0$ is
\begin{equation}
    \frac{2^n-k}{2^n-k+1} \ldots \frac{2^n-1}{2^n} = \frac{2^n - k}{2^n}.
\end{equation}
Hence, with $k = (1/4)2^n$, the probability that the first $k$ queries are all zero is $(3/4)$. In such an event, the classical algorithm will have no information that helps it to distinguish the remaining $(3/4) 2^n$ point functions.
Because the conditional probability of the distribution over $a$ under such an event is uniform across the $(3/4) 2^n$ point functions, no matter what the classical algorithm chooses, the probability of correctly identifying $a$ is equal to
$(4/3)2^{-n}$.
Thus the probability of correctly identifying a uniformly random label $a \in \left\{ 0,1\right\}^n$ with $k$ uniformly random queries $f_a (x_1),\ldots, f_a (x_k)$ is equal to
\begin{equation}
    \frac{1}{4} + \frac{3}{4} \cdot \frac{4}{3} \frac{1}{2^n} \leq \frac{1}{2}
\end{equation}
under the extra assumption $n \geq 2$ ($2^{n} \geq 4$).
Combining this fundamental lower bound with the constructive argument above reveals $N \geq (1/4)2^n=\Omega (2^n)$.
The number of inputs $N$ must be greater than $k = (2^n) / 4$ to achieve a success probability of at least $2/3$.
\end{proof}

It is not hard to consider a subset of all quantum states that can be mapped to the basic problem of learning point functions.
We can equate $x \in \{I, X, Y, Z\}^n$ (input) and $a \in \left\{I,X,Y,Z\right\}^n$
with bit strings of size $2n$ (two bits unambiguously characterize all 4 single-qubit Paulis).
We recall the definition that $P_x \in \{I, X, Y, Z\}^{\otimes n}$ is the tensor product of Pauli basis based on $x \in \{I, X, Y, Z\}^n$. For each $a$, we define a mixed state
\begin{equation}
    \rho_a = \frac{I + P_a}{2^n}, \quad \text{such that} \quad  \Tr(Z_1 \cE_{\rho_a}(\ketbra{x}{x})) = \Tr(P_x \rho_a) = \begin{cases}
    1 & \text{if } x = a,\\
    0 & \text{if } x \neq a.\\
    \end{cases}
\end{equation}
This is now exactly the same as the problem for learning point function given in Lemma~\ref{lem:learnpointf} with input size $2n$.
Hence if the restricted classical ML model can produce $f(x)$ such that
\begin{equation}
    \max_{x \in \{I, X, Y, Z\}^n} \left| f(x) - \Tr(Z_1 \cE_\rho(\ketbra{x}{x})) \right| < \frac{1}{2}\quad \text{with probability at least $2/3$},
\end{equation}
 for all $a \in \{I, X, Y, Z\}^n$ from $N$ data samples. Then, the number of inputs must obey $N \geq (1/4)2^{2n}= \Omega(4^n)$.

\subsection{Sample complexity lower bound for any classical ML algorithm}
\label{sec:lowerboundindmeas}

While restricted classical ML models can only obtain measurement outcome $o$ for a fixed observable $O$, one may be curious what the sample complexity lower bound would be for a standard classical ML model that can perform an arbitrary POVM measurement (which is equivalent to performing quantum computation with ancilla qubits follow by a computational basis measurement) on the output quantum state $\cE_{\rho}(\ketbra{x}{x})$.
While this is much more powerful than restricted classical ML models, we will show that the sample complexity is still exponential.
A quantum ML model that can process the quantum data in an entangled fashion has an exponential advantage over classical ML model that can only process each quantum data separately.

We will first focus on non-adaptive measurements, where the POVM measurement for each copy is fixed and do not change throughout the training process, and show that any procedure needs to measure $\Omega(n 2^n)$ copies of the state $\rho$.
A matching upper bound can be obtained using classical shadows with random Clifford measurements \cite{huang2020predicting}. Classical shadows with random Clifford measurements can predict $M$ observables $O_1, \ldots, O_M$ using only $N_{\mathrm{C}} = \mathcal{O}(\max_{i} \Tr(O_i^2)\log(M))$ copies.
For the set of all $n$-qubit Pauli observables, we have $\Tr(P_i^2) = 2^n$ and $M = 4^n$, so $N_{\mathrm{C}} = \mathcal{O}(n 2^n)$.
This matches with the lower bound for non-adaptive measurements.

We also give a lower bound of $\Omega(2^{n / 3})$ for adaptive measurements, where each POVM measurement can depend on the outcomes of previous POVM measurements. This sample complexity lower bound could be further improved using a more sophisticated analysis, which we leave for future work.

\subsubsection{Non-adaptive measurements}

When one could perform arbitrary POVM measurement on $\cE_{\rho}(\ketbra{x}{x})$, the input $x$ is no longer useful since $x$ only rotates the quantum state $\rho$, which can be absorbed into the POVM measurement. Let us denote the POVM for the $i$-th copy of $\rho$ to be $F_i$. We can reduce the classical ML models with non-adaptive measurements to the following setup:
\begin{equation}
    \rho \xrightarrow{F_1} o_1, \,\, \ldots, \,\, \rho \xrightarrow{F_N} o_N,
\end{equation}
where $o_i$ is the POVM measurement outcome (single shot). Hence, it is a random variable that depends on $\rho$ and $F_i$.
This setup is known as \emph{single-copy independent measurements} in the quantum state tomography literature \cite{haah2017sample, huang2020predicting}.
Without loss, we can further restrict our attention to POVM measurements comprised of rank-one projectors. Such measurements always reveal more information and we write
$F_i = \{w_{i o_i} 2^n \ketbra{\psi_{i o_i}}{\psi_{i o_i}}\}_j$.
The classical ML model then uses the classical measurement outcomes $o_1, \ldots, o_N$ to learn a function $f(x)$ such that the following prediction error bound holds with high probability:
\begin{equation}
    \max_{x \in \{I, X, Y, Z\}^n} \left| f(x) - \Tr(Z_1 \cE_\rho(\ketbra{x}{x})) \right| = \max_{x \in \{I, X, Y, Z\}^n} \left| f(x) - \Tr(P_x \rho) \right| < \frac{1}{2}.
\end{equation}
We will show that this necessarily requires $N \geq \Omega(n 2^n)$.
Recall that the sample complexity for achieving a constant worst-case prediction error using a quantum ML model is $N = \mathcal{O}(n)$.

The proof uses a mutual information analysis similar to the sample complexity lower bound for quantum ML models given in Section~\ref{sec:MIanalysisML}.
We consider a communication protocol between Alice and Bob.
First, we define a codebook that Alice will use to encode classical information in quantum states:
\begin{equation} \label{eq:rhoa-def}
    (a,s) \in \left\{1, \ldots, 4^n - 1\right\} \times  \left\{\pm 1 \right\} \,\,\, \longrightarrow \,\,\, \rho_{(a,s)} = \frac{I + s P_a}{2^n},
\end{equation}
where $P_a$ runs through all Pauli matrices  $\{I, X, Y, Z\}^{\otimes n} \setminus  \{I^{\otimes n}\}$ that are not the global identity and $s$ is an additional sign.
There are $2 (4^n-1)$ combinations in total and
Alice will sample one of them uniformly at random.
Then, Alice she prepares $N$ copies of $\rho_{(a,s)}$ and sends them to Bob.
Bob will perform the POVM measurements $F_1, \ldots, F_N$ on the $N$ copies of $\rho_{(a,s)}$ he receives to obtain classical measurement outcomes $o_1, \ldots, o_N$.
He will use them to train the classical ML model to produce  a function $\tilde{f}(x)$ that is guaranteed to obey
\begin{equation} \label{eq:learningguarantee}
    \max_{x \in \{I, X, Y, Z\}^n} \left| \tilde{f}(x) - \Tr(Z_1 \cE_\rho(\ketbra{x}{x})) \right| = \max_{x \in \{I, X, Y, Z\}^n} \left| \tilde{f}(x) - \Tr(P_x \rho_{(a,s)}) \right| < \frac{1}{2}
\end{equation}
with high probability.
Because $\Tr(P_x \rho_{(a,s)})$ is either $+1, 0, -1$, it is not hard to see that Bob can use $\tilde{f}(x)$ to determine Alice's original message $a$ -- as long as Equation~\eqref{eq:learningguarantee} holds.
Using Fano's inequality and data processing inequality, the mutual information between $a$ and the measurement outcomes $o_1, \ldots, o_N$ must be lower bounded by
\begin{equation}
    I({\small (a,s)} : o_1,\ldots, o_N) = \Omega(\log(2(4^n - 1))) = \Omega(n).
\end{equation}
Furthermore, when conditioned on $a$, the measurement outcomes $o_1, \ldots, o_N$ are all independent from each other. Hence we can use the chain rule of mutual information to obtain
\begin{equation}
    I({\small (a,s)} : o_1, \ldots, o_N) \leq \sum_{i=1}^N I({\small (a,s)} : o_i).
\end{equation}
By construction of $\rho_a$, we can show that $I(a : o_i) \leq 1 / (2^n + 1)$. This is the content of Lemma~\ref{lem:expsmallI} below. This technical result allows us to conclude
\begin{equation}
    \frac{N}{2^n + 1} \geq I( {\small (a,s)} : o_1, \ldots, o_N) = \Omega(n).
\end{equation}
This establishes the advertised result: $N =\Omega(n 2^n)$.

The bound on mutual information is a nontrivial consequence of the following technical statement.

\begin{lemma} \label{lem:pauli-averages}
Fix a pure state $|\psi \rangle \! \langle \psi|$ and
set $\rho_{a,s} = (I + s P_a)/2^n$, where $P_a$ is chosen uniformly from $\{I, X, Y, Z\}^{\otimes n} \setminus \{ I^{\otimes n}\}$ and $s \in \left\{ \pm 1 \right\}$ is a random sign. Then,
\begin{align}
\E_{a,s} \langle \psi| \rho_{(a,s)} |\psi \rangle =& \frac{1}{2^n} \quad \text{and} \\
\E_{a,s} \langle \psi| \rho_{(a,s)} |\psi \rangle^2 =& \frac{1}{4^n}\left( 1 + \frac{1}{4^n-1}\sum_a \langle \psi|P_a |\psi \rangle^2 \right) = \frac{1}{4^n} \left( 1 + \frac{1}{2^n+1}\right).
\end{align}
\end{lemma}

\begin{proof}
The first display is an immediate consequence of symmetry:
\begin{align}
\E_{a,s} \langle \psi | \rho_{(a,s)} |\psi \rangle = \frac{1}{2^n} \E_{a,s} \langle \psi| I |\psi \rangle + \frac{1}{2^n}\E_a\left( \E_s s \langle \psi| P_a |\psi \rangle \right)
= \frac{1}{2^n} + 0.
\end{align}
The second display follows from the fact that the collection $\left\{2^{-n/2} P_a \right\}$ forms an orthonormal basis of the space of all Hermitian $2^n \times 2^n$ matrices (with respect to the Hilbert-Schmidt inner product). Parseval's identity then asserts
$
\frac{1}{2^n} (\sum_a \langle \psi| P_a |\psi \rangle^2 + \langle \psi |I| \psi \rangle^2)
= \| |\psi \rangle \! \langle \psi \|_2^2 = 1
$ and, together with symmetry, we conclude
\begin{align}
\E_{a,s} \langle \psi| \rho_{(a,s)}|\psi \rangle^2
=& \frac{1}{4^n} \E_{a,s} \langle \psi| I |\psi \rangle^2 + \frac{2}{4^n} \E_a \left( \E_s s \langle \psi| P_a |\psi \rangle \right) + \frac{1}{4^n}\E_a (\E_s s^2 )\langle \psi| P_a |\psi \rangle^2 \\
=& \frac{1}{4^n} + 0 + \frac{1}{4^n(4^n-1)}\sum_a \langle \psi| P_a |\psi \rangle^2
=\frac{1}{4^n} + \frac{1}{2^n (4^n-1)}\left( 1- \frac{1}{2^n}\right) \\
=&
\frac{1}{4^n}\left( 1+ \frac{1}{ (2^n+1)} \right)
\end{align}
\end{proof}

We now give the desired upper bound on the mutual information between $a$ and $o_i$.
\begin{lemma} \label{lem:expsmallI}
$I(a : o_i) \leq 1 / (2^n + 1).$
\end{lemma}
\begin{proof}
Suppose that the POVM measurement $F_i$ is given by $F_{o_i} = \{w_{i o_i} 2^n \ketbra{\psi_{i o_i}}{\psi_{i o_i}}\}$, where $j$ ranges through all possible measurement outcomes.
Using the condition $\sum_{j} w_{i o_i} 2^n \ketbra{\psi_{i o_j}}{\psi_{i o_i}} = I$ and $\braket{\psi_{i o_j}|\psi_{i o_i}} = 1$, we have $\sum_{o_i} w_{i o_i} = 1$ (take the trace of both sides of the equation)
The probability distribution of $o_i$ conditioned on $(a,s)$ can hence be written as
\begin{equation}
    p_{(a,s)}(o_i) = w_{i o_i} 2^n \bra{\psi_{i o_i}} \rho_{(a,s)} \ket{\psi_{i o_i}}
\end{equation}
and the mutual information obeys
\begin{align}
    \label{eq:MI-upper-bound}
    I(\small{(a,s)} : o_i) =& \E_{a,s}\left[ \sum_{o_i} p_{(a,s)}(o_i) \log p_{(a,s)}(o_i) \right] - \sum_{o_i} \left(\E_{a,s} p_{(a,s)}(o_i)\right)  \log\left(\E_{a,s} p_{(a,s)}(o_i)\right) \\
    \leq& \sum_{o_i} \frac{\E_{a,s} \left[ p_{(a,s)}(o_i)^2\right] - \left(\E_{a,s} p_{(a,s)}(o_i)\right)^2 }{\E_{a,s} p_{(a,s)}(o_i)}.
\end{align}
The inequality uses the fact that $\log(x)$ is concave, so $\log(x) \leq \log(y) + \frac{x-y}{y}$.
We can separately compute $\E_{a,s} p_{(a,s)}(o_i)$ and $\E_{a,s} \left[ p_{(a,s)}(o_i)^2\right]$ using Lemma~\ref{lem:pauli-averages}:
\begin{align}
\E_{a,s} p_{(a,s)} (o_i) =& w_{i o_i} 2^n \E_{a,s} \langle \psi_{i o_i}| \rho_{(a,s)}|\psi_{i o_i} \rangle
= w_{i o_i}, \\
\E_{a,s} p_{(a,s)}(o_i)^2
=& w_{i o_i}^2 4^n \E_{a,s} \langle \psi_{i o_i}| \rho_{(a,s)}|\psi_{i o_i}\rangle^2 = w_{i o_i}^2 \left(1+\frac{1}{2^n+1}\right)
\end{align}
Inserting these expressions into
Equation~\eqref{eq:MI-upper-bound} reveals
\begin{equation}
    I(a : o_i) \leq
    \sum_{o_i} \frac{1}{w_{i o_i}}\left(w_{i o_i}^2 \left(1 + \frac{1}{2^n+1} \right) - w_{i o_i}^2 \right)
    =
    \frac{1}{2^n+1}\sum_{o_i} w_{i o_i} = \frac{1}{2^n + 1},
\end{equation}
because the $w_{io_i}$'s are expansion coefficients of a POVM ($\sum_{o_i} w_{i o_i} = 1$).
This is the advertised result.
\end{proof}

\subsubsection{Adaptive measurements}

In the last section, we have derived a sample complexity lower bound for independent single-copy quantum measurements. Several results in the literature address this restricted setting, see e.g.\ \cite{haah2017sample,huang2020predicting}.
In stark contrast, very little is known about the more realistic setting of adaptive single-copy measurements; see \cite{bubeck2020entanglement, aharonov2021quantum} for lower bounds on quantum mixedness testing and the task of distinguishing between physical experiments.
Here, we give an elementary proof that provides an exponential sample complexity lower bound for predicting Pauli expectation values based on single-copy adaptive measurements.
Such an extension to adaptive measurement strategies is nontrivial -- very few results are known for this setting.
However, the actual result is not (yet) tight.
We believe that further improvements are possible using more sophisticated analysis and we leave this as a future work.
When adaptive measurements are allowed, each POVM measurement can depend on previous POVM measurement outcomes. And the entire training process can be written as
\begin{equation}
    \rho \xrightarrow{F^{o_{ < 1}}_1} o_1, \rho \xrightarrow{F^{o_{ < 2}}_2} o_2, \,\, \ldots, \,\, \rho \xrightarrow{F^{o_{ < N}}_N} o_N.
\end{equation}
Here, $o_i$ is the POVM measurement outcome of the $i$-th measurement and $o_{ < k}$ subsumes all previous measurement outcomes (think $o_{<k} = \{ o_1,\ldots,o_{k-1}\}$).
Without loss, we again restrict our attention to POVM measurements comprised of rank-one projectors (these always provide more information in the measurement outcome) and write
\begin{equation}
    F^{o_{< i}}_i = \left\{ w^{o_{< i}}_{ij} 2^n \ketbra{\psi^{o_{< i}}_{ij}}{\psi^{o_{< i}}_{ij}} \right\}.
\end{equation}
Using the conditions $\sum_{j} w^{o_{< i}}_{ij} 2^n \ketbra{\psi^{o_{< i}}_{ij}}{\psi^{o_{< i}}_{ij}} = I$ and $\langle \psi^{o_{< i}}_{ij}| \psi^{o_{< i}}_{ij} \rangle=1$, we have $\sum_{j} w^{o_{< i}}_{ij} = 1$.

Suppose Alice randomly chooses one of $4^n$ possible $n$-qubit states based on a \emph{non-uniform} probability distribution:
\begin{equation}
    \rho =
    \begin{cases}
    \rho_{\mathrm{mm}} = \frac{I}{2^n}, & \text{with probability} \,\, \frac{1}{2},\\
    \rho_{a} = \frac{I + P_a}{2^n}, & \text{with probability} \,\, \frac{1}{2(4^n - 1)},
    \end{cases}
\end{equation}
where $P_a$ is chosen among all $4^n - 1$ nontrivial Pauli operators $\{I, X, Y, Z\}^{\otimes n} \setminus \{I^{\otimes n}\}$. (In contrast to the previous subsection, we only include positive signs, that is $\rho_a = \rho_{(a,+1)}$).
Alice then sends $\rho^{\otimes N}$ to Bob.
Hence, Bob will receive
\begin{equation}
    \begin{cases}
    \rho_{\mathrm{mm}}^{\otimes N}, & \text{with probability} \,\, \frac{1}{2},\\
    \rho_{a}^{\otimes N}, & \text{with probability} \,\, \frac{1}{2 (4^n - 1)}.
    \end{cases}
\end{equation}
Bob uses the classical ML model with adaptive measurement outcomes $o_1, \ldots, o_N$ to infer all Pauli expectation values of $\Tr(P_b \rho)$ with a small error (the assumption of the classical ML model).
Because $\Tr(P_b \rho_a) = \delta_{ab}$, but $\Tr(P_b \rho_{\mathrm{mm}}) = 0$ for all $b$, Bob can successfully distinguish the quantum state chosen by Alice once he knows the expectation values of all Pauli operators.
This implies that the probability distribution for Bob's measurement outcomes $o_1, \ldots, o_N$ must be able to distinguish the two events:
\begin{enumerate}
    \item $\rho = \rho_{\mathrm{mm}}$, which happens with probability $1/2$.
    \item $\rho = \rho_a$ for a random $a$, which happens with probability $1/2$.
\end{enumerate}
We have thus reduced a multiple-hypothesis testing problem -- distinguishing among $4^n$ possible states -- to a two-hypothesis testing problem -- distinguishing between the completely mixed state and a randomly chosen $\rho_a$.
We will use this observation to derive an information-theoretic lower bound on $N$.

Importantly, the two hypothesises give rise to different joint probability distributions of all $N$ outcomes.
Under the first hypothesis ($\rho=\rho_{\mathrm{mm}}$),
\begin{equation}
    p_1(o_1, \ldots, o_N) = \prod_{i=1}^N w^{o_{< i}}_{i o_i}, \label{eq:p1}
\end{equation}
while the second hypothesis ($\rho=\rho_a$ and $a$ is itself uniformly random) would imply
\begin{equation}
    p_2(o_1, \ldots, o_N) = \frac{1}{4^n - 1} \sum_{P_a: \{I, X, Y, Z\}^{\otimes n} \setminus \{I^{\otimes n}\}} \prod_{i=1}^N w^{o_{< i}}_{i o_i} (1 + \bra{\psi^{o_{< i}}_{i o_i}} P_a\ket{\psi^{o_{< i}}_{i o_i}}). \label{eq:p2}
\end{equation}
When Bob performs his measurement strategy, he exactly obtains one sample from such a joint probability distribution. And based on this sample, he must distinguish between the two hypotheses. A well-known fact from statistics states that the optimal decision strategy is the maximum likelihood rule (pick the joint probability distribution that is most likely, given the observed event). This strategy succeeds with a probability that is determined by the total variational (TV)distance:
\begin{align}
\mathrm{Pr} \left[ \text{successful discrimination} \right] =& \frac{1}{2}+\frac{1}{2} \mathrm{TV}(p_1, p_2)
\quad \text{with} \\
\mathrm{TV}(p_1, p_2) =&  \frac{1}{2}\sum_{x \in \mathcal{X}} \left|p_1(x) - p_2(x)\right|    =\sum_{x \in \mathcal{X}: p_1(x) > p_2(x)} \left(p_1(x) - p_2(x)\right). \label{eq:TV-def}
\end{align}
We note in passing that this classical observation is actually the starting point for the celebrated Holevo-Helstrom theorem \cite{Helstrom1969,Holevo1973}. We refer to \cite{ambainis_wise_2007,Matthews_2009} and also \cite[Lecture 1]{Kueng2019} for a modern discussion from a quantum information perspective.

Importantly, the TV distance between $p_1 (o_1,\ldots,o_N)$ and $p_2 (o_1,\ldots,o_N)$ remains tiny until $N$ becomes exponentially large: $\mathrm{TV}\left(p_1 (o_1,\ldots,o_N), p_2 (o_1,\ldots,o_N)\right) \leq 2N / (2^n+1)^{1/3}$. This is the content of Lemma~\ref{lem:total-variation} below.
This TV upper bounds Bob's bias for successful discrimination of the two possibilities.
Because Bob can successfully discriminate between the two hypotheses, we have $\mathrm{TV}(p_1,p_2)=\Omega (1)$, which gives the desired result
\begin{equation}
N  = \Omega ( 2^{n/3}).
\end{equation}
We emphasize that an exponential lower bound can be proven using the same method whenever we want to predict a class of observables $\{O_1, \ldots, O_M\}$ such that $\frac{1}{M} \sum_{i=1}^M \bra{\psi} O_i \ket{\psi}^2$ is an exponentially small number for all pure states $\ket{\psi}$.
Pauli observables are one example of such a class of observables.

\begin{lemma} \label{lem:total-variation}
$\mathrm{TV}\left(p_1 (o_1,\ldots,o_N), p_2 (o_1,\ldots,o_N)\right) \leq 2N / (2^n+1)^{1/3}$.
\end{lemma}

\begin{proof}
The key insight is that -- regardless of the actual choice of measurements -- single-copy, rank-one POVMs are ill equipped to distinguish the maximally mixed state $\rho_{\mathrm{mm}}= I/2^n$ from $\rho_a=(I+P_a)/2^n$, where $a$ is a uniformly random index.
To see this, let us first re-use Lemma~\ref{lem:pauli-averages} to obtain
\begin{equation}
 \sum_a \bra{\psi^{o_{< i}}_{i o_i}} P_a \ket{\psi^{o_{< i}}_{i o_i}}^2
 = (4^n-1) \E_a \bra{\psi^{o_{< i}}_{i o_i}} P_a \ket{\psi^{o_{< i}}_{i o_i}}^2 = \frac{4^n-1}{2^n+1}=2^n-1 \quad \text{for any $\ket{\psi^{o_{< i}}_{i o_i}}$}. \label{eq:Pauli-sum}
\end{equation}
This, in turn implies that for an overwhelmingly large fraction of Pauli indices, $\bra{\psi^{o_{< i}}_{i o_i}} P_a \ket{\psi^{o_{< i}}_{i o_i}}$ must be exponentially close to zero. More precisely,
\begin{equation}
    \left|\left\{ P_a \in \{I, X, Y, Z\}^{\otimes n} \setminus \{I^{\otimes n}\}: |\bra{\psi^{o_{< i}}_{i o_i}} P_a\ket{\psi^{o_{< i}}_{i o_i}}| \leq \frac{1}{(2^{n} + 1)^{1/3}} \right\}\right| \geq \left(1 - \frac{1}{(2^{n} + 1)^{1/3}}\right) (4^n - 1), \label{eq:small-paulis}
\end{equation}
and this cardinality bound can be proven by contradiction.
Suppose that the number of very small Pauli operators is smaller than $\left(1 - \frac{1}{(2^{n} + 1)^{1/3}}\right) (4^n - 1)$. Then, this would necessarily imply
\begin{equation}
    \sum_{P_a: \{I, X, Y, Z\}^{\otimes n}} \bra{\psi^{o_{< i}}_{i o_i}} P_a\ket{\psi^{o_{< i}}_{i o_i}}^2 > \frac{4^n - 1}{(2^{n} + 1)^{1/3}} \times \left(\frac{1}{(2^{n} + 1)^{1/3}}\right)^2 = 2^n - 1
\end{equation}
which is in direct conflict with Equation~\eqref{eq:Pauli-sum}.

Equation~\eqref{eq:small-paulis} states that an overwhelming fraction of all Pauli observables $P_a$ have small overlap with a single fixed rank-one projector. This feature makes the associated states $\rho_a$ difficult to distinguish from the maximally mixed state and an adaptive measurement procedure cannot easily wash out this phenomenon. Fix a tuple $(o_1,\ldots,o_N)$ of measurement outcomes and let
\begin{equation}
G^{(o_1,\ldots,o_N)} = \left\{ P_a \in \{I, X, Y, Z\}^{\otimes n} \setminus \{I^{\otimes n}\}: |\bra{\psi^{o_{< i}}_{i o_i}} P_a\ket{\psi^{o_{< i}}_{i o_i}}| \leq \frac{1}{(2^{n} + 1)^{1/3}}, \forall i =1, \ldots, N \right\}
\end{equation}
denote the set of Pauli operators that have exponentially small overlap with the associated rank-one projectors.
Note that the set $G^{(o_1,\ldots,o_N)}$ depends on the measurement outcomes $(o_1,\ldots,o_N)$.
Then, the cardinality of this set follows from the lower bound of the size of the set given in Equation~\eqref{eq:small-paulis}:
\begin{equation*}
|G^{(o_1,\ldots,o_N)}| \geq \left( 1 - \frac{N}{(2^n+1)^{1/3}} \right) \left(4^n-1\right)
\end{equation*}
and $G^{(o_1,\ldots,o_N)}$ exclusively contains mixed states $\rho_a$ that are difficult to distinguish from the maximally mixed state. This has profound implications on the TV distance.
Using the definition of $p_1, p_2$ and the triangle inequality, we obtain
\begin{equation}\label{eq:TV-Pauli-sum}
    \mathrm{TV}(p_1, p_2) \leq \frac{1}{4^n - 1} \sum_{\substack{o_{1:N}:\\ p_1(o_{1:N}) > p_2(o_{1:N})}} \sum_{P_a: \{I, X, Y, Z\}^{\otimes n} \setminus \{I^{\otimes n}\}} \left(\prod_{i=1}^N w^{o_{< i}}_{i o_i}\right) \left( 1 - \prod_{i=1}^N (1 + \bra{\psi^{o_{< i}}_{i o_i}} P_a\ket{\psi^{o_{< i}}_{i o_i}})\right)
\end{equation}
and each expression on the very right is guaranteed to be contained in an (exponentially) small interval
\begin{equation}\label{eq:bound-Pauli-product}
\left( 1 - \prod_{i=1}^N (1 + \bra{\psi^{o_{< i}}_{i o_i}} P_a\ket{\psi^{o_{< i}}_{i o_i}})\right) \in \left[1-2^N,1 \right].
\end{equation}
To obtain an upper bound on the TV distance we make an additional rounding argument: if $P_a \notin G^{(o_1,\ldots,o_N)}$, we assume that the difference takes the maximum possible value $1 - \prod_{i=1}^N (1 + \bra{\psi^{o_{< i}}_{i o_i}} P_a\ket{\psi^{o_{< i}}_{i o_i}}) = 1$.
We also used the short-hand notation $o_{1:N} = (o_1,\ldots,o_N)$.
This gives
\begin{align}
    \mathrm{TV}(p_1, p_2) &\leq \frac{1}{4^n - 1} \sum_{\substack{o_{1:N}:\\ p_1(o_{1:N}) > p_2(o_{1:N})}} |G^{(o_1,\ldots,o_N)}| \left(\prod_{i=1}^N w^{o_{< i}}_{i o_i}\right) \left( 1 - \left( 1 - \frac{1}{(2^n+1)^{1/3}}\right)^N\right)\\
    &+ \frac{1}{4^n - 1} \sum_{\substack{o_{1:N}:\\ p_1(o_{1:N}) > p_2(o_{1:N})}} (4^n - 1 - |G^{(o_1,\ldots,o_N)}|) \left(\prod_{i=1}^N w^{o_{< i}}_{i o_i}\right)\\
    &\leq \frac{1}{4^n - 1} \sum_{\substack{o_{1:N}:\\ p_1(o_{1:N}) > p_2(o_{1:N})}} (4^n - 1) \left(\prod_{i=1}^N w^{o_{< i}}_{i o_i}\right) \left( 1 - \left( 1 - \frac{1}{(2^n+1)^{1/3}}\right)^N\right)\\
    &+ \frac{1}{4^n - 1} \sum_{\substack{o_{1:N}:\\ p_1(o_{1:N}) > p_2(o_{1:N})}} \frac{N}{(2^{n} + 1)^{1/3}} (4^n - 1) \left(\prod_{i=1}^N w^{o_{< i}}_{i o_i}\right)\\
    &=  \sum_{\substack{o_{1:N}:\\ p_1(o_{1:N}) > p_2(o_{1:N})}} \left(\prod_{i=1}^N w^{o_{< i}}_{i o_i}\right) \left( \left( 1 - \left( 1 - \frac{1}{(2^n+1)^{1/3}}\right)^N\right) + \frac{N}{(2^{n} + 1)^{1/3}}\right)\\
    &\leq \left( 1 - \left( 1 - \frac{1}{(2^n+1)^{1/3}}\right)^N\right) + \frac{N}{(2^{n} + 1)^{1/3}}.
\end{align}
To obtain the first inequality, we divide the sum over $P_a$ in Equation~(\ref{eq:TV-Pauli-sum}) into a sum over $P_a\in G^{(o_1,\ldots,o_N)}$ and a sum over $P_a\not\in G^{(o_1,\ldots,o_N)}$, and apply the upper bound in Equation~(\ref{eq:bound-Pauli-product}) to the sum over $P_a\not\in G^{(o_1,\ldots,o_N)}$.
The second inequality uses the fact $\left(1 - \frac{N}{(2^n-1)^{1/3}}\right)(4^n-1) \leq |G^{(o_1,\ldots,o_N)}| \leq 4^n - 1$.
The final line follows from $w^{o_{< i}}_{i o_i} \geq 0$ and $\sum_{o_i} w^{o_{< i}}_{i o_i} = 1$; hence
\begin{equation}
    \sum_{\substack{o_{1:N}:\\ p_1(o_{1:N}) > p_2(o_{1:N})}} \left(\prod_{i=1}^N w^{o_{< i}}_{i o_i}\right) \leq 1.
\end{equation}
Finally, we use the inequality $1 - (1-x)^N \leq Nx, \forall x \leq 1, N \in \mathbb{N}$ to find
\begin{equation}
    \mathrm{TV}(p_1, p_2) \leq \frac{2N}{(2^{n} + 1)^{1/3}}.
\end{equation}
This concludes the desired upper bound on the total variation distance.
\end{proof}

\subsection{Sample complexity lower bound for general entangled measurements}
\label{sec:lowerboundentangled}

We establish a lower bound of $N = \Omega(n)$ for quantum ML models that can perform general entangled measurements on $N$ copies of $\rho$. This matches with the sample complexity of the quantum ML model considered in Section~\ref{sec:pauli-qml}.
The proof is similar to the sample complexity lower bound for single-copy independent measurements given in Section~\ref{sec:lowerboundindmeas}.

We consider a communication protocol between Alice and Bob.
First, we define a codebook that Alice will use to encode classical information in quantum states:
\begin{equation} \label{eq:rhoa-def-2}
    a \in \left\{1, \ldots, 4^n - 1\right\}\; \text{and} \; s \in \left\{\pm 1 \right\} \,\,\, \longrightarrow \,\,\, \rho_{(a,s)} = \frac{I + s P_a}{2^n},
\end{equation}
where $P_a$ runs through all Pauli matrices  $\{I, X, Y, Z\}^{\otimes n}\} \setminus  \{I^{\otimes n}\}$ that are not the global identity.
Alice samples a Pauli index $a$ and a sign $s$ uniformly at random.
Then, she prepares $N$ copies of $\rho_a$ and sends them to Bob.
Bob will use the quantum ML model on the $N$ copies of $\rho$ to predict the expectation values of all $4^n$ Pauli observables.
The expectation values will be a general entangled POVM measurement outcome $o$ on the $N$ copies of $\rho$.
By assumption, Bob can use $O$ to construct a function $\tilde{f}(x): \{I, X, Y, Z\}^n \rightarrow \mathbb{R}$ that is guaranteed to satisfy
\begin{equation} \label{eq:learningguarantee-Q}
    \max_{x \in \{I, X, Y, Z\}^n} \left| \tilde{f}(x) - \Tr(Z_1 \cE_\rho(\ketbra{x}{x})) \right| = \max_{x \in \{I, X, Y, Z\}^n} \left| \tilde{f}(x) - \Tr(P_x \rho_{(a,s)}) \right| < \frac{1}{2}
\end{equation}
with high probability.
Because $\Tr(P_x \rho_{(a,s)})$ is either $+1, 0, -1$, it is not hard to see that Bob can use $\tilde{f}(x)$ to determine both $a$ and $s$ as long as Equation~\eqref{eq:learningguarantee} holds.
Using Fano's inequality and data processing inequality, the mutual information between $(a,s)$ and the measurement outcome $o$ must be lower bounded by
\begin{equation}
    I({\small (a,s)} : o) \geq \Omega(\log(2(4^n - 1))) = \Omega(n).
\end{equation}
In conjunction with Holevo's theorem \cite{Helstrom1969,Holevo1973}, we have
\begin{align}
    I({\small (a,s)} : o) &\leq S\left(\frac{1}{2 (4^n - 1)} \sum_{a,s} \rho_{(a,s)}^{\otimes N}\right) - \frac{1}{2 (4^n - 1)} \sum_{a,s} S\left(\rho_{(a,s)}^{\otimes N}\right)\\
    &\leq N n \log(2) - \frac{1}{2 (4^n - 1)} \sum_{a,s} N S\left( \rho_{(a,s)} \right) \\
    &= N n \log(2) - N (n-1) \log(2) = N \log(2),
\end{align}
where $S(\cdot)$ is the von Neumann entropy, the second inequality uses the fact a $Nn$-qubit system has an entropy upper bounded by $Nn \log(2)$.
Hence, we must have $N = \Omega(n)$.

\end{document}